\documentclass[a4paper]{amsart}
\usepackage{amssymb}
\usepackage{amsmath}
\usepackage{amscd}
\usepackage{amsthm}
\usepackage{verbatim}
\usepackage{enumerate}

\allowdisplaybreaks
\newcommand{\FF}{\mathcal{F}}
\newcommand{\MM}{\mathcal{M}}
\newcommand{\NN}{\mathcal{N}}
\newcommand{\EE}{\mathcal{E}}
\newcommand{\PP}{\mathcal{P}}
\newcommand{\BB}{\mathcal{B}}
\newcommand{\dd}{\partial}
\newcommand{\ra}{\rightarrow}
\newcommand{\hra}{\hookrightarrow}

\newcommand{\ola}{\overleftarrow}
\newcommand{\RR}{\mathbb{R}}
\newcommand{\be}{\begin{equation}}
\newcommand{\ee}{\end{equation}}

\newcommand{\id}{\mathrm{id}}

\newcommand{\HH}{\mathcal{H}}
\newcommand{\DD}{\mathcal{D}}
\newcommand{\Fun}{\mathrm{Fun}}
\newcommand{\ZZ}{\mathbb{Z}}
\newcommand{\bt}{\bullet}
\newcommand{\mr}{\mathrm}
\newcommand{\Maps}{\mathrm{Maps}}
\newcommand{\Map}{\mathrm{Map}}
\newcommand{\End}{\mathrm{End}}
\newcommand{\lifted}{\mathrm{lifted}}
\newcommand{\ad}{\mathrm{ad}}
\newcommand{\ev}{\mathrm{ev}}
\newcommand{\g}{\mathfrak{g}}
\newcommand{\h}{\mathfrak{h}}
\newcommand{\p}{\mathsf{p}}

\newcommand{\Vect}{\mathfrak{X}}
\newcommand{\ver}{\mathrm{vert}}
\newcommand{\hor}{\mathrm{hor}}
\newcommand{\aux}{\mathrm{aux}}

\newcommand{\kin}{\mathrm{kin}}
\newcommand{\target}{\mathrm{target}}
\newcommand{\Diff}{\mathrm{Diff}}
\newcommand{\BV}{\mathrm{BV}}
\newcommand{\A}{\mathcal{A}}
\newcommand{\z}{\mathrm{z}}
\newcommand{\tr}{\mathrm{tr}}
\newcommand{\Conf}{\mathrm{Conf}}

\newcommand{\Pexp}{\mathcal{P}\exp}
\newcommand{\mb}{\mathbf}
\newcommand{\ii}{\mathsf{i}}
\newcommand{\iii}{i}
\newcommand{\ddd}{\delta}
\newtheorem{thm}{Theorem}
\newtheorem{prop}{Proposition}
\newcommand{\LL}{\mathcal{L}}
\newcommand{\tot}{\mathrm{tot}}

\theoremstyle{remark}
\newtheorem{rem}{Remark}

\theoremstyle{definition}
\newtheorem{definition}{Definition}

\begin{document}
\title{A construction of observables for AKSZ sigma models}

\author{Pavel Mnev}
\address{Institut f\"ur Mathematik,
Universit\"at Z\"urich,
Winterthurerstrasse 190,
CH-8057, Z\"urich, Switzerland}

\email{pmnev@pdmi.ras.ru}

\maketitle

\begin{abstract}
A construction of gauge-invariant observables is suggested for a class of topological field theories, the AKSZ sigma-models. The observables are associated to extensions of the target $Q$-manifold of the sigma model to a $Q$-bundle over it with additional Hamiltonian structure in fibers.
\end{abstract}

\tableofcontents

\section{Introduction}

The interest in observables in topological field theories is largely due to the applications in algebraic topology, where expectation values of observables are known to yield (at least in some examples) invariants of knots and links under ambient isotopy or, more generally, cohomology classes of spaces of embeddings. The seminal example here was given in the work of Witten \cite{Witten_Jones}, where the expectation value of the Wilson loop observable in  Chern-Simons theory with gauge group $SU(2)$, associated to a knot in the 3-sphere, was found out to yield the Jones polynomial of the knot.

Given a diffeomorphism-invariant action $S_\Sigma$ of a topological field theory on a manifold $\Sigma$ with space of fields $F_\Sigma$, one is particularly interested in observables $O\in C^\infty(F_\Sigma)$ associated to embedded submanifolds $\ii:\gamma\hra \Sigma$ which depend on fields only via their pull-back from $\Sigma$ to $\gamma$:
\be O_{\gamma,\ii}(X)=O(\ii^*X) \label{O = O(i^* X)}\ee
where $X$ is the field. Such observables are automatically invariant with respect to diffeomorphisms $\phi: \Sigma\ra\Sigma$ in the sense that
\be O_{\gamma,\phi\circ \ii}((\phi^{-1})^*X)=O_{\gamma,\ii}(X) \label{O diff-invar}\ee
where $(\phi^{-1})^*X$ is the pull-back of the field by $\phi^{-1}$.

On a formal level, the expectation value of the observable
\be \langle O_{\gamma,\ii} \rangle = \int_{F_\Sigma} \DD X\; O_{\gamma,\ii}(X)\; e^{i S_\Sigma(X)} \label{O PI}\ee
is diffeomorphism-invariant
\be \langle O_{\gamma,\phi\circ \ii} \rangle = \langle O_{\gamma,\ii} \rangle \label{diff-invar of <O>}\ee
due to (\ref{O diff-invar}),  diffeomorphism-invariance of the action $S_\Sigma$ and of the path integral measure $\DD X$. This means in particular that the expectation value is an invariant of the embedding $\ii:\gamma\hra \Sigma$ under ambient isotopy.

Since topological field theories possess gauge symmetry, one also requires that the observable is gauge-invariant, so that in the path integral (\ref{O PI}) one could pass to integration over gauge equivalence classes of fields.

In this paper we employ Batalin-Vilkovisky (BV) formalism to treat systems with gauge symmetry (cf. e.g. \cite{Schwarz} for details). In particular the space of fields $F_\Sigma$ becomes extended to an odd-symlectic $Q$-manifold $\FF_\Sigma$ with the action $S_\Sigma$ extended to a function on $\FF_\Sigma$, satisfying the master equation\footnote{In this outline we do not make explicit distinction between classical and quantum master equation. We write the classical one which also coincides with the quantum one if $\Delta_\Sigma S=0$ for $\Delta_\Sigma$ the BV Laplacian. Likewise, we write the classical gauge-invariance condition on the observable $Q_\Sigma O=0$ which coincides with the quantum one if $\Delta_\Sigma O=0$, cf. section \ref{sec: pre-obs and obs} for details.} $\{S_\Sigma,S_\Sigma\}=0$ and generating the cohomological vector field $Q_\Sigma$ as its Hamiltonian vector field. In this formalism the gauge-invariance of the observable is expressed as
\be Q_\Sigma O=0 \label{QO=0}\ee
and the path integral (\ref{O PI})
for the expectation value is traded for an integral over a Lagrangian submanifold $\LL$ in $\FF_\Sigma$, with $\LL$ playing the role of the gauge-fixing condition.

One idea how to construct an observable for a gauge theory in BV formalism is to consider an extension of the BV theory $(\FF_\Sigma,S_\Sigma)$ to a BV theory on a larger space $(\FF_\Sigma\times \FF^\aux, S_\Sigma+S^\aux)$ with $\FF^\aux$ the space of ``auxiliary fields'' and $S^\aux$ the action for auxiliary fields $Y$ which is also allowed to depend on the ``ambient'' fields $X\in \FF_\Sigma$. Then one uses the BV push-forward construction
\be O(X)=\int_{\LL^\aux\subset \FF^\aux}\DD Y\; e^{i S^\aux(X,Y)} \label{BV push-forward}\ee
(propositions \ref{prop: pre-obs to obs}, \ref{prop: quantum pre-obs to obs} in section \ref{sec: pre-obs and obs}, cf. also \cite{DiscrBF}) to integrate out the auxiliary fields and produce an observable for the ambient theory $(\FF_\Sigma,S_\Sigma)$. In this paper we call such extensions by auxiliary fields ``pre-observables'' (to be more precise, we use a slightly better definition with 
equation (\ref{pre-observable CME}) on $S^\aux$ instead of the master equation on $S_\Sigma+S^\aux$, which still suffices to produce an observable, cf. definition \ref{def: class pre-observable} and remark \ref{rem: pre-obs to CME for S+S^aux}).

In this paper we consider observables obtained by the construction outlined above for topological sigma models coming from AKSZ construction \cite{AKSZ} (cf. a brief reminder in section \ref{sec: AKSZ reminder}), where $\FF_\Sigma$ is the mapping space from the tangent bundle of $\Sigma$ with shifted parity of fibers to a target $Q$-manifold $\MM$ with additional Hamiltonian structure (cf. definition \ref{def: Ham Q-mfd}). In this setting we give a construction of pre-observables (section \ref{sec: AKSZ pre-observable}), which associates a pre-observable for the AKSZ sigma model to an embedded submanifold $\ii:\gamma\hra \Sigma$ and an extension of the target $\MM$ to a $Q$-bundle over $\MM$ with additional Hamiltonian structure in fibers (the ``Hamiltonian $Q$-bundle'', cf. section \ref{sec: Ham Q-bun}).
For the ambient and auxiliary actions to have degree zero, degrees of Hamiltonian structures in the base and the fiber of the target $Q$-bundle have to match\footnote{More precisely, degrees of Hamiltonians (not of symplectic structures) have to match the dimensions.} dimensions of $\Sigma$ and $\gamma$ respectively.

Having a pre-observable for an AKSZ sigma model, we construct the corresponding observable by the BV push-forward (\ref{BV push-forward}). By construction, this observable is gauge-invariant in the sense of (\ref{QO=0}) and satisfies (\ref{O = O(i^* X)}), and so formally produces expectation values which are invariant under ambient isotopy. Of course, the problem with this definition of the observable is that (\ref{BV push-forward}) is generally a path integral. In section \ref{sec: from pre-obs to obs} we consider several simple cases when this integral can be made sense of and its expected properties can be rigorously checked:
\begin{itemize}
\item case when the fiber in the target is a point (which gives rise to observables given by an exponential of an integral of a local expression),
\item case of one-dimensional observables (when the path integral becomes a path integral of quantum mechanical type and can be regularized via geometric quantization, which gives rise to certain generalization of Wilson loops),
\item case when the integral (\ref{BV push-forward}) is Gaussian which gives rise to a class of observables which we call ``torsion-like'' for their similarity to Ray-Singer torsion.
\end{itemize}

In section \ref{sec: examples of observables} we give explicit examples of observables falling within one of the three classes above. In particular, we recover the usual Wilson loop observable in Chern-Simons theory, with the corresponding path integral expression being the Alekseev-Faddeev-Shatashvili path integral formula for the Wilson loop \cite{AFS}. We also recover the Cattaneo-Rossi ``Wilson surface'' observable for $BF$ theory \cite{Cattaneo-Rossi}.

In this paper we concentrate only on the construction of observables, we are not calculating their expectation values.

\subsection{Logic of the construction}
\begin{enumerate}[(i)]
\item \label{logic step 1} For an AKSZ sigma model on manifold $\Sigma$ with target $\MM$, find a Hamiltonian $Q$-bundle $\EE$ over $\MM$. Then for every embedded submanifold $\ii:\gamma\hra \Sigma$ of fixed dimension matching the degree of Hamiltonian structure in fibers of $\EE$, we construct a pre-observable.
\item \label{logic step 2}
Take the fiber BV integral (\ref{BV push-forward}) over the auxiliary fields of
the pre-observable constructed in step (\ref{logic step 1})  to obtain an observable for the AKSZ sigma model.
\end{enumerate}
Unlike step (\ref{logic step 1}), step (\ref{logic step 2}) is not canonical, in the sense that it is a quantization problem: the path integral has to be made sense of (and the gauge-invariance of the result has to be checked) which can be done for certain classes of target $Q$-bundles, but it is not clear whether it is possible to do in greater generality\footnote{One can indeed try to define the path integral perturbatively, as sum of Feynman diagrams represented by integrals over compactified configuration spaces of tuples of points on $\gamma$. However, for the formal argument of proposition \ref{prop: pre-obs to obs} that the result is gauge-invariant to become rigorous, one has to prove that the hidden boundary strata of the configuration spaces appearing in the calculation of $Q_\Sigma O$ (cf. the proof of proposition \ref{prop 6}) do not contribute, which is not true in general case.}.

\subsection{Plan of the paper}
The logical organization of the exposition is as follows.
\begin{itemize}
\item Sections \ref{sec: AKSZ reminder}, \ref{sec: Q-bun reminder}, bits of \ref{sec: pre-obs and obs} concerning BV observables --- background reminders.
\item Section \ref{sec: AKSZ pre-observable} --- step (\ref{logic step 1}) of the main construction.
\item Section \ref{sec: Ham Q-bun} --- auxiliary construction for step (\ref{logic step 1}) of the main construction.
\item Section \ref{sec: pre-obs and obs} --- motivation for step (\ref{logic step 2}) of the main construction.
\item Section \ref{sec: from pre-obs to obs} --- examples for step (\ref{logic step 2}) of the main construction.
\item Section \ref{sec: examples of observables} --- fully explicit examples.
\end{itemize}

Section \ref{sec: AKSZ reminder} is a short reminder of the AKSZ construction of topological sigma models in Batalin-Vilkovisky formalism. We also recall how some well-known topological field theories fit into the construction: Chern-Simons theory, $BF$ theory, Poisson sigma model.

In section \ref{sec: Ham Q-bun} we first briefly recall the standard notion of a $Q$-bundle and then define a ``Hamiltonian'' $Q$-bundle, preparing the grounds for the construction of pre-observables in section \ref{sec: AKSZ pre-observable}. We start with the definition \ref{def: trivial Ham Q-bundles} of a trivial Hamiltonian $Q$-bundle (``trivial'' here means trivial as a fiber bundle; the cohomological vector field is \textit{not} required to be a sum of a cohomological vector field on the base and another one on the fiber). All our examples are of this type, but for the completeness of the exposition, we also give a slightly more general definition \ref{def: general Ham Q-bundle}, which does not explicitly rely on the total space being a direct product. However in a local trivialization definition \ref{def: general Ham Q-bundle} boils down to definition \ref{def: trivial Ham Q-bundles}.

In section \ref{sec: pre-obs and obs} we recall the standard notions of classical and quantum observables in BV formalism and introduce the notion of a pre-observable, which comes in three modifications:
\begin{enumerate}[(i)]
\item Classical pre-observable, definition \ref{def: class pre-observable}: essentially, an extension of the action of the ambient classical BV theory to a solution of classical master equation on the space of fields extended by auxiliary fields. We give a technically more convenient definition with equation (\ref{pre-observable CME}) required
    instead of the classical master equation on extended space of fields (cf. remark \ref{rem: pre-obs to CME for S+S^aux} for the relation between the two).
\item Semi-quantum pre-observable, definition \ref{def: semi-quantum pre-observable}, suited for integrating out auxiliary fields to obtain an observable for the ambient theory (proposition \ref{prop: pre-obs to obs}).
\item Quantum pre-observable for a quantum ambient BV theory (i.e. one with the action satisfying the quantum master equation), definition \ref{def: quantum observable}: this is also an extension of the ambient theory by auxiliary fields, plus an extension of the action to a solution of quantum master equation on the extended space. From a quantum pre-observable one can induce a quantum observable for the ambient theory, by integrating out auxiliary fields (proposition \ref{prop: quantum pre-obs to obs}).
\end{enumerate}
In this section in the discussion of quantization we always work with spaces of fields as with finite dimensional spaces and the integrals are over finite-dimensional super-manifolds. In the context of local field theory, the BV push-forward becomes a path integral, so the proofs given in section \ref{sec: pre-obs and obs} stop to work and the propositions become conjectures that have to be proven by more delicate means (cf. the proofs of propositions \ref{prop 5}, \ref{prop 6} in section \ref{sec: from pre-obs to obs}).

Section \ref{sec: AKSZ pre-observable} is the logical core of the paper. Here we give the construction of a classical pre-observable for an AKSZ theory out of an extension of the target to a Hamiltonian $Q$-bundle.

In section \ref{sec: from pre-obs to obs} we treat several classes of situations when the BV push-forward yielding an observable out of a pre-observable constructed in section \ref{sec: AKSZ pre-observable} can be performed rigorously: case of target fiber being a point, case of 1-dimensional observables (via geometric quantization), case when the BV push-forward is given by a Gaussian integral.

In section \ref{sec: examples of observables} we specialize the constructions of section \ref{sec: from pre-obs to obs} to present several explicit examples, including the Wilson loop together with its path integral representation known from \cite{AFS}, Cattaneo-Rossi codimension 2 observable for $BF$ theory, torsion observables in Chern-Simons and $BF$ theories, etc.

\subsection{Acknowledgements} I wish to thank Anton Alekseev, Alberto Cattaneo, Andrei Losev and Nicolai Reshetikhin for inspiring discussions.
This work was partially supported by RFBR grant  11-01-00570-a and by SNF grant 200021\_137595.

\section*{Terminology, notations}
\textbf{Terminology.}
\begin{itemize}
\item $Q$-manifold (or bundle) = differential graded manifold (or bundle).
\item Ghost number = internal degree (to distinguish from de Rham degree of a differential form) = $\ZZ$-grading on functions on $\ZZ$-graded manifolds.
\item Observable = gauge-invariant functional on the space of fields (cf. section \ref{sec: pre-obs and obs} for definitions).
\item Expectation value = correlator.
\end{itemize}

\textbf{Conventions.}
We set the Planck's constant $\hbar=1$ (cf. remark \ref{rem: Planck's constant} on how to re-introduce $\hbar$).

\textbf{Notations.}
We use $\MM,\NN,\EE,\FF$ etc. for $\ZZ$-graded manifolds, 
$M,N$ etc. for ordinary (non-graded) manifolds; $\Sigma$ always denotes the source (spacetime) manifold of the sigma model, $\gamma$ is typically a submanifold of $\Sigma$ on which the observable is supported. We use $\LL$ for a Lagrangian submanifold of a degree -1 symplectic graded manifold.

We denote by $\Vect(\MM)$ the Lie algebra of vector fields on $\MM$. For a fiber bundle $\pi: \EE\ra\MM$, we denote by $\Vect^\ver(\EE)$ the Lie algebra of vertical vector fields on the total space $\EE$, i.e. the space of sections $\Gamma(\EE,T^\ver\EE)$ of the vertical distribution on $\EE$, $T^\ver\EE=\ker(d\pi)\subset T\EE$.

We denote the degree (the ghost number) of functions/differential forms/vector fields on a graded manifold by $|\cdots|$.

\section{AKSZ reminder}
\label{sec: AKSZ reminder}
What follows is a very short reminder of the AKSZ construction of topological sigma models in Batalin-Vilkovisky formalism, to fix the terminology and notation. We refer the reader to the original paper \cite{AKSZ} and the later expositions in \cite{CF_AKSZ}, \cite{Royt} for details.
\subsection{Target data.} Let $\MM$ be a degree $n$ sympectic $Q$-manifold, i.e. a $\ZZ$-graded manifold endowed with a degree 1 vector field $Q$ satisfying $Q^2=0$ (the cohomological vector field) and with
a degree\footnote{By ``degree'' here we mean the internal $\ZZ$-grading coming from the grading on $\MM$.} $n$ symplectic form $\omega\in\Omega^2(\MM)$ which is compatible with $Q$, i.e. $L_Q\omega=0$.

Assume\footnote{In fact (cf. \cite{Royt}), for $n\neq -1$ the symplectic property of the cohomological vector field ($L_{Q}\omega=0$) implies existence and uniqueness of the Hamiltonian $\Theta=\frac{1}{n+1}\,\iota_E \iota_Q\omega$ where $E$ is the Euler vector field. Maurer-Cartan equation (\ref{Theta MC}) follows from $Q^2=0$ for $n\neq -2$. Also, for $n\neq 0$, a closed form $\omega$ is automatically exact.}
that $Q$ has a Hamiltonian function $\Theta\in C^\infty(\MM)$ with  $|\Theta|=n+1$, $\{\Theta,\bt\}_{\omega}=Q$ and satisfying
\be\{\Theta,\Theta\}_{\omega}=0 \label{Theta MC}\ee
Also assume that $\omega$ is exact, with $\alpha\in \Omega^1(\MM)$ a primitive.

\begin{definition}\label{def: Ham Q-mfd}
We call the set of data $(\MM,Q,\omega=\delta\alpha,\Theta)$ a \textit{Hamiltonian $Q$-manifold of degree $n$}.
\end{definition}

\subsection{The AKSZ sigma model.} Fix a Hamiltonian $Q$-manifold $\MM$ of degree $n\geq -1$ and let $\Sigma$ be an oriented closed manifold, $\dim\Sigma=n+1$.
Then one constructs the space of fields as the space of graded maps between graded manifolds from the degree-shifted tangent bundle $T[1]\Sigma$ to $\MM$:
\be \FF_\Sigma=\Map(T[1]\Sigma,\MM) \label{mapping space}\ee
It is a $Q$-manifold with the cohomological vector field coming from the lifting of $Q$ on the target and of the de Rham operator $d_\Sigma$ on $\Sigma$ (viewed as a cohomological vector field on $T[1]\Sigma$) to the mapping space:
\be Q_\Sigma=(d_\Sigma)^{\lifted}+(Q)^{\lifted} \qquad \in \Vect(\FF_{\Sigma}) \label{Q AKSZ}\ee

\textit{Transgression map.} The following natural maps
\be
\begin{CD}
\FF_{\Sigma}\times T[1]\Sigma @>\ev>> \MM \\
@VpVV \\
\FF_{\Sigma}
\end{CD}
\label{ev-p diagram}
\ee
(where $p$ is the projection to the first factor) allow us to define the transgression map
\be \tau_{\Sigma}= p_*\ev^*:\quad \Omega^\bt(\MM)\ra \Omega^\bt(\FF_{\Sigma}) \label{AKSZ transgression}\ee
Here $p_*$ is the fiber integration over $T[1]\Sigma$ (with the canonical integration measure). Map $\tau_{\Sigma}$ preserves the de Rham degree of a form, but changes the internal grading (the ghost number) by $-\dim\Sigma$.

If $u\in \Vect(T[1]\Sigma)$, $v\in \Vect(\MM)$ are any vector fields on the source and the target and $\psi\in \Omega^\bt(\MM)$ is a form on the target, then for the Lie derivatives of the transgressed form along the lifted vector fields we have
\begin{eqnarray}
L_{u^\lifted}\tau_\Sigma(\psi) &=& 0 \label{L_v tau = 0} \\
L_{v^\lifted} \tau_\Sigma(\psi) &=& (-1)^{|v|\dim\Sigma} \tau_\Sigma(L_v\psi) \label{L_u tau = tau L_u}
\end{eqnarray}
The finite version of (\ref{L_v tau = 0}) is that for $\Phi:T[1]\Sigma\ra T[1]\Sigma$ a diffeomorphism of the source,
we have
\be (\Phi^*)^*\tau_\Sigma(\psi)= \tau_\Sigma(\psi) \label{Phi^** tau}\ee
where $\Phi^*: \FF_\Sigma\ra\FF_\Sigma$ is the lifting of $\Phi$ to the mapping space and $(\Phi^*)^*:\Omega^\bt(\FF_\Sigma)\ra \Omega^\bt(\FF_\Sigma)$ is the pull-back by $\Phi^*$.

\textit{The BV 2-form and the master action.} One obtains the degree $-1$ symplectic form (the ``BV 2-form'') on $\FF_{\Sigma}$ from the target by transgression\footnote{We introduce the sign $(-1)^{\dim\Sigma}$ in this definition to avoid signs in the formula for the action below. The reader may encounter different sign conventions for the AKSZ construction in the literature.}:
\be \Omega_{\Sigma}=(-1)^{\dim\Sigma}\tau_{\Sigma}(\omega) \quad \in \Omega^2(\FF_{\Sigma}) \label{AKSZ Omega}\ee
The Hamiltonian function for $Q_{\Sigma}$ (the master action) is constructed as
\be S_{\Sigma} = \underbrace{\iota_{d_\Sigma^{\lifted}}\;\tau_{\Sigma}(\alpha)}_{S_{\Sigma}^\mr{kin}} + \underbrace{\tau_{\Sigma}(\Theta)}_{S_{\Sigma}^\mr{target}}\qquad \in C^\infty (\FF_{\Sigma}) \label{AKSZ S}\ee
It automatically satisfies the classical master equation
$$\{S_{\Sigma},S_{\Sigma}\}_{\Omega_{\Sigma}}=0$$

One can summarize the construction above by saying that we have a degree $-1$ Hamiltonian $Q$-manifold structure on the mapping space (\ref{mapping space}): $(\FF_\Sigma,Q_\Sigma,\Omega_\Sigma,S_\Sigma)$. The primitive $1$-form for the BV 2-form can also be constructed by transgression as $\tau_\Sigma(\alpha)$.

Note that exact shifts of the target 1-form, $\alpha\mapsto \alpha + \delta f$, with $f\in C^\infty(M)$ leave $S_\Sigma$ unchanged.

\textit{Why AKSZ theory is topological.} For $\phi\in \Diff(\Sigma)$ a diffeomorphism of $\Sigma$, denote $\tilde\phi\in \Diff(T[1]\Sigma)$ the tangent lift of $\phi$ to $T[1]\Sigma$. Then
\be (\tilde\phi^*)^*S_\Sigma=S_\Sigma,\qquad (\tilde\phi^*)^*\Omega_\Sigma=\Omega_\Sigma \label{AKSZ action diff-invar}\ee
because of (\ref{Phi^** tau}) and because $d_\Sigma\in\Vect(T[1]\Sigma)$ commutes with $\tilde\phi$, since the latter is a tangent lift.

\textit{In coordinates.} Let $x^a$ be local homogeneous coordinates on the target $\MM$, let $u^\mu$ be local coordinates on $\Sigma$ and $\theta^\mu=du^\mu$ be the associated degree 1 fiber coordinates on $T[1]\Sigma$. Then locally an element of $\FF_\Sigma$ 
is parameterized by
\be X^a(u,\theta)=\sum_{k=0}^{\dim\Sigma}\;\underbrace{\sum_{1\leq \mu_1<\cdots<\mu_k\leq\dim\Sigma} X^a_{\mu_1\cdots\mu_k}(u)\;\theta^{\mu_1}\cdots\theta^{\mu_k} }_{X^a_{(k)}(u,\theta)} \label{superfield}\ee
Coefficient functions $X^a_{\mu_1\cdots\mu_k}(u)$ are local coordinates of degree $|x^a|-k$ on the mapping space $\FF_\Sigma$.  Expression (\ref{superfield}) is known as (the component of) the \textit{superfield}, and it can be regarded as a generating function for the coordinates on the mapping space $\FF_\Sigma$.

For any function $f\in C^\infty(\MM)$, we have
\begin{multline}
\ev^*f=f(X)=\\
=f(X_{(0)})+X^a_{\geq 1}\cdot (\dd_a f)(X_{(0)})+\frac{1}{2} X^a_{\geq 1} X^b_{\geq 1}\cdot (\dd_b \dd_a f)(X_{(0)}) +\cdots\\
\in C^\infty(\FF_\Sigma\times T[1]\Sigma)
\end{multline}
where we denoted $X^a_{\geq 1}=\sum_{k\geq 1}X^a_{(k)}$ the part of the superfield of positive de Rham degree with respect to $\Sigma$; $\ev$ is the horizontal arrow in (\ref{ev-p diagram}).

Let $\alpha$ and $\omega$ be locally given as
$\alpha=\alpha_a(x)\, \delta x^a$ and $\omega=\frac{1}{2}\,\omega_{ab}(x)\,\delta x^a \wedge \delta x^b$.
Then the BV 2-form and its primitive are: 
$$\Omega_\Sigma=(-1)^{\dim\Sigma}\int_\Sigma \frac{1}{2}\;\omega_{ab}(X)\;\delta X^a \wedge \delta X^b,\qquad \alpha_\Sigma=\int_\Sigma \alpha_a(X)\;\delta X^a $$
(Note that we use $\delta$ to denote the de Rham differential on the target $\MM$ and on the mapping space $\FF_\Sigma$. We reserve symbol $d$ for the de Rham differential on the source $\Sigma$.)

The master action is:
$$S_\Sigma(X)=\underbrace{\int_\Sigma \alpha_a(X)\; dX^a}_{S^\mr{kin}_\Sigma} + \underbrace{\int_\Sigma \Theta(X)}_{S^\mr{target}_\Sigma}$$

If the cohomological vector field on the target is locally written as $Q=Q^a(x)\,\frac{\dd}{\dd x^a}$ then the cohomological vector field (\ref{Q AKSZ}) is determined by its action on the components of the superfield:
$$Q_\Sigma X^a=d X^a + Q^a(X)$$

Critical points of $S_\Sigma$ are (with our sign conventions) $Q$-anti-morphisms between $T[1]\Sigma$ and $\MM$, i.e. $Q$-morphisms between $(T[1]\Sigma,d)$ and $(\MM,-Q)$.

\subsection{Examples}
\label{sec:AKSZ examples}
Here we recall some of the standard examples of the AKSZ construction.

\textit{Chern-Simons theory \cite{AKSZ}.} Let $\g$ be a quadratic Lie algebra, i.e. a Lie algebra with a non-degenerate invariant pairing $(,)$. Denote by $\psi: \g[1]\ra \g$ the degree $1$ $\g$-valued coordinate on $\g[1]$. We choose the target Hamiltonian $Q$-manifold of degree $2$ as
\begin{multline} \label{CS target}
\MM=\g[1],\quad Q=\left\langle \frac{1}{2}[\psi,\psi],\frac{\dd}{\dd\psi} \right\rangle,\\ \omega=\frac{1}{2}(\delta\psi,\delta\psi),\quad \alpha=\frac{1}{2}(\psi,\delta\psi),\quad
\Theta=\frac{1}{6} (\psi,[\psi,\psi])
\end{multline}
where $\langle,\rangle$ is the canonical pairing\footnote{We will generally be using notation $\langle,\rangle$ for the canonical pairing $V\otimes V^* \ra \RR$ between a vector space $V$ and its dual; sometimes we will indicate the respective $V$ as a subscript: $\langle,\rangle_V$.} between $\g$ and $\g^*$.
The associated AKSZ sigma model on a closed oriented 3-manifold $\Sigma$ has the space of fields
$$\FF_\Sigma=\Map(T[1]\Sigma,\g[1])\cong \g[1]\otimes \Omega^\bt(\Sigma)$$
The superfield is
\be A=A_{(0)}+A_{(1)}+A_{(2)}+A_{(3)} \label{CS superfield}\ee
with $A_{(k)}$ a coordinate on $\FF_\Sigma$ with values in $\g$-valued $k$-forms on $\Sigma$, with internal degree (ghost number) $1-k$, for $k=0,1,2,3$. The BV 2-form is
$$\Omega_\Sigma=-\frac{1}{2}\int_\Sigma (\delta A,\delta A)$$
and the action is
$$S_\Sigma=\int_\Sigma \frac{1}{2}(A,dA)+\frac{1}{6}(A,[A,A]) $$
This is the action of Chern-Simons theory in Batalin-Vilkovisky formalism.

In case $\g=\RR$ with abelian Lie algebra structure, we have $Q=\Theta=0$ on the target and
\be \FF_\Sigma=\Omega^\bt(\Sigma)[1],\quad \Omega_\Sigma=-\frac{1}{2}\int_\Sigma \delta A\wedge \delta A,\quad S_\Sigma=\frac{1}{2}\int_\Sigma A\wedge dA \label{abelian CS}\ee
This is the \textit{abelian} Chern-Simons theory in BV formalism.

\textit{$BF$ theory.} For $\g$ a Lie algebra (not necessarily quadratic\footnote{However, for the consistency of quantization, in particular for the quantum master equation (\ref{QME}), one has to require that $\g$ is unimodular.}) and $D$ a non-negative integer, we define the target Hamiltonian $Q$-manifold of degree $D-1$ as
\begin{multline} \label{BF target}
\MM=\g[1]\oplus \g^*[D-2],\quad Q=\left\langle \frac{1}{2}[\psi,\psi],\frac{\dd}{\dd\psi} \right\rangle+\left\langle \ad^*_\psi\xi, \frac{\dd}{\dd\xi}\right\rangle,\\ \omega=\langle\delta\xi,\delta\psi\rangle,\quad \alpha=\langle \xi,\delta\psi \rangle,\quad
\Theta=\frac{1}{2} \langle \xi,[\psi,\psi] \rangle
\end{multline}
Here $\psi: \g[1]\ra \g$ is as before and $\xi: \g^*[D-2]\ra \g^*$ is the $\g^*$-valued coordinate on $\g^*[D-2]$ of degree $D-2$; $\ad^*$ is the coadjoint action of $\g$ on $\g^*$.

The associated AKSZ sigma model on a closed oriented $D$-manifold $\Sigma$ has the space of fields
\be \FF_\Sigma=\g[1]\otimes \Omega^\bt(\Sigma) \oplus \g^*[D-2]\otimes \Omega^\bt(\Sigma) \label{BF F}\ee
The superfields associated to $\psi$ and $\xi$ are respectively
$$A=\sum_{k=0}^D A_{(k)},\quad B=\sum_{k=0}^D B_{(k)}$$
with $A_{(k)}$ a $\g$-valued $k$-form on $\Sigma$ of internal degree $1-k$; $B_{(k)}$ is a $\g^*$-valued $k$-form of internal degree $D-2-k$. The BV 2-form and the action are:
\be \Omega_\Sigma=(-1)^D \int_\Sigma \langle \delta B,\delta A \rangle,\quad S_\Sigma=\int_\Sigma \left\langle B, dA+\frac{1}{2}[A,A] \right\rangle  \label{BF Omega and S}\ee
This is the $BF$ theory in BV formalism.

In abelian case, $\g=\RR$, we have $Q=\Theta=0$ on the target and
\begin{multline}
\FF_\Sigma=\Omega^\bt(\Sigma)[1]\oplus \Omega^\bt(\Sigma)[D-2],\\
\Omega_\Sigma=(-1)^D \int_\Sigma \delta B\wedge \delta A,\quad S_\Sigma=\int_\Sigma B\wedge dA \label{abelian BF}
\end{multline}

\textit{Poisson sigma model \cite{CF_AKSZ}.} Let $M$ be a manifold endowed with a Poisson bivector $\pi\in \Gamma(M,\wedge^2TM)$. We construct the target Hamiltonian $Q$-manifold of degree $1$ as
\begin{multline}
\label{PSM target data}
\MM=T^*[1]M,\\
Q=\left\langle \pi(x), p\wedge \frac{\dd}{\dd x} \right\rangle_{\wedge^2 T_x M}+ \frac{1}{2}\left\langle \frac{\dd}{\dd x} \pi(x),(p\wedge p) \otimes \frac{\dd}{\dd p} \right\rangle_{(\wedge^2 T\otimes T^*)_x M}, \\
\omega=\langle\delta p,\delta x\rangle, \quad \alpha=\langle p,\delta x \rangle,\quad
\Theta=\frac{1}{2}\left\langle \pi(x), p\wedge p \right\rangle_{\wedge^2 T_x M}
\end{multline}
Here $x$ and $p$ stand for the local base and fiber coordinates on $T^*[1]M$ respectively. Note that all objects in (\ref{PSM target data}) are globally well defined.

The corresponding AKSZ sigma model on an oriented closed surface $\Sigma$ has the space of fields
$$\FF_\Sigma=\Map(T[1]\Sigma,T^*[1]M)$$
with superfields
$$X=X_{(0)}+X_{(1)}+X_{(2)},\quad \eta=\eta_{(0)}+\eta_{(1)}+\eta_{(2)}$$
associated to local coordinates $x$ and $p$ on the target, respectively. Here $X_{(k)}$ and $\eta_{(k)}$ are $k$-forms on $\Sigma$ with internal degrees $-k$ and $1-k$ respectively. The BV 2-form and the action are:
$$\Omega_\Sigma=\int_\Sigma \langle \delta \eta , \delta X \rangle, \quad
S_\Sigma=\int_\Sigma \langle \eta, dX\rangle+\frac{1}{2}\langle \pi(X),\eta\wedge\eta \rangle$$

\section{Hamiltonian $Q$-bundles}\label{sec: Ham Q-bun}
In this section we briefly recall the standard notion of a $Q$-bundle and then introduce ``Hamiltonian $Q$-bundles'', an auxiliary notion necessary for our construction of observables.
\subsection{$Q$-bundles reminder} \label{sec: Q-bun reminder}
Recall that a ``$Q$-bundle'' \cite{Kotov-Strobl} is a fiber bundle in the category of $Q$-manifolds.

In particular, a \textit{trivial} $Q$-bundle is trivial bundle of graded manifolds
$$\pi:\underbrace{\MM\times \NN}_{\EE} \ra \MM$$
where the base $\MM$ is endowed with a cohomological vector field $Q$ and the total space is endowed with a cohomological vector field $Q^{\tot}$ in such a way that $\pi$ is a $Q$-morphism, i.e. $d\pi(Q^\tot)=Q$. This implies the following ansatz for $Q^\tot$:
\be Q^\tot=Q+\A \label{Q_E decompostition}\ee
where $\A\in \Vect^\mr{vert}(\EE)\cong\Vect(\NN)\hat\otimes C^\infty(\MM)$ is the vertical part of $Q^\tot$. Cohomological property for $Q^\tot$ is equivalent to
\be Q\A+ \frac{1}{2}[\A,\A]=0 \label{A flatness}\ee
plus the cohomological property for $Q$. Note that in the first term on the l.h.s. of (\ref{A flatness}) we are thinking of $Q$ as acting on $C^\infty(\MM)$ part of $\A$. Equivalently, we can lift $Q$ to a horizontal vector field on $\EE$ and write the first term as $[Q,\A]$.

\subsection{Trivial Hamiltonian $Q$-bundles}
\label{sec: trivial Ham Q-bundles}
\begin{definition}\label{def: trivial Ham Q-bundles}
We define a trivial Hamiltonian $Q$-bundle of degree $n'\in\ZZ$ as the following collection of data.
\begin{enumerate}[(i)]
\item A trivial $Q$-bundle
$$\pi: \EE=\MM\times \NN\ra \MM$$
with $Q^\tot=Q+\A$ as in (\ref{Q_E decompostition}).
\item The fiber $\NN$ is endowed with a degree $n'$ exact symplectic form $\omega'=\delta\alpha'$ and a degree $(n'+1)$ Hamiltonian $\Theta'\in C^\infty(\EE)$ satisfying
$\A=\{\Theta',\bt\}_{\omega'}$ and
\be Q \Theta'+\frac{1}{2}\{\Theta',\Theta'\}_{\omega'}=0 \label{fiber CME}\ee
\end{enumerate}
\end{definition}

\begin{rem}(Roytenberg) 
\label{rem: symp form is exact}
Symplectic structure $\omega'$ is automatically exact for $n'\neq 0$, since $\omega'=\delta\left(\frac{1}{n'}\iota_{E}\omega'\right)$ where $E$ is the Euler vector field on $\NN$. 
\end{rem}

\begin{rem}\label{rem: B-S integrality}
If $n'=0$, then $\omega'$ is not automatically exact. In this situation, exactness condition may be relaxed\footnote{Cf. remark \ref{rem: S_kin for omega' integral} for the motivation: primitive $\alpha'$ (or connection $\nabla'$ in the relaxed version) will be needed in the construction of section \ref{sec: AKSZ pre-observable} to define the kinetic part of the auxiliary action.}
to the Bohr-Sommerfeld 
condition that $\omega'/2\pi$ defines an integral cohomology class in $H^2(\NN)$.
Then the role of primitive 1-form $\alpha'$ is taken by a Hermitian line bundle $L'$ over $\NN$ equipped with a $U(1)$-connection $\nabla'$ of curvature $\omega'$.

\end{rem}

\begin{rem} For a Hamiltonian $Q$-manifold of degree $n\neq -2$, equation $\{\Theta,\Theta\}=0$ follows from the fact that $\Theta$ is the Hamiltonian for a cohomological vector field (Jacobi identity implies $\{\frac{1}{2}\{\Theta,\Theta\},\bt\}=Q^2(\bt)=0$ which implies that $\{\Theta,\Theta\}$ is a constant of degree $n+2$ and thus vanishes unless $n=-2$).

Observe that this argument fails for the fiber $\NN$: applying the same reasoning to derive the equation (\ref{fiber CME}) from (\ref{A flatness}), we obtain that l.h.s. of (\ref{fiber CME}) is a pull-back of some function on the base of degree $n'+2$, which does not imply that it is zero.
\end{rem}

\subsection{General Hamiltonian $Q$-bundles}
\label{sec: general Ham Q-bundles}
The assumption of triviality of the bundle $\EE$, used in the definition \ref{def: trivial Ham Q-bundles}, can be relaxed as follows.
\begin{definition} \label{def: general Ham Q-bundle}
A Hamiltonian $Q$-bundle is the following set of data:
\begin{enumerate}[(i)]
\item A $Q$-bundle $\pi:\EE\ra\MM$.
\item A degree $n'$ exact pre-symplectic form $\omega'=\delta\alpha'$ on the total space $\EE$, with the property that the distribution $\ker\omega'\subset T\EE$ is transversal to the vertical distribution $T^\ver\EE$; thus $\ker\omega'$ defines a flat Ehresmann connection $\nabla_{\omega'}$ on $\EE$.
\item A Hamiltonian function $\Theta'\in C^\infty(\EE)$ satisfying
$$\iota_{Q^\tot}\omega'=\delta^\ver \Theta'$$
where $\delta^\ver=(i^\ver)^*\circ \delta$ is the vertical part of de Rham differential on $\EE$; we denoted $i^\ver: T^\ver\EE\hra T\EE$ the natural inclusion. We also require that
$$(Q^\hor+\frac{1}{2}Q^\ver)\Theta'=0$$
where the splitting $Q^\tot=Q^\hor+Q^\ver$ of the cohomological vector field on $\EE$ into the horizontal and the vertical part is defined by connection $\nabla_{\omega'}$.
\end{enumerate}
\end{definition}

\begin{rem} Note that one can introduce a local trivialization of $\EE$ consistent with the flat connection $\nabla_{\omega'}$
(i.e. where the horizontal distribution $\ker\omega'\subset T\EE$ is defined by the direct product structure on $\EE|_U$ coming from the local trivialization, where $U\subset \MM$ is a trivializing neighborhood).
In such a trivialization, we recover back the definition \ref{def: trivial Ham Q-bundles}. In this sense, definition \ref{def: general Ham Q-bundle} does not bring in a drastic increase of generality.
\end{rem}

\subsection{Examples}
\label{sec: examples of Ham Q-bun}
\begin{enumerate}[(i)]
\item \label{ex: quadr Lie algebra with a orthog module}Let $\g$ be a 
Lie algebra 
and let $\rho:\g\ra \mathfrak{so}(R)\subset \End(R)$ be a representation of $\g$ on by anti-symmetric matrices on a Euclidean vector space $R,(,)$. Then for any $k\in\ZZ$ we can define a degree $4k+2$ trivial Hamiltonian $Q$-bundle with base $\MM=\g[1]$ and fiber $\NN=R[2k+1]$. Let $\psi$ be the degree 1 $\g$-valued coordinate on $\g[1]$ and let $x$ be the degree $2k+1$ $R$-valued coordinate on $R[2k+1]$. Then we set
    \be \label{g[1] structure}
    Q=\frac{1}{2}\left\langle[\psi,\psi],\frac{\dd}{\dd \psi}\right\rangle 
    \ee
    -- the Chevalley-Eilenberg differential, and
    \begin{multline} \label{module structure}
    \A=\left\langle\rho(\psi) x, \frac{\dd}{\dd x}\right\rangle,\; \omega'=\frac{1}{2}(\delta x,\delta x),\\ \alpha'=\frac{1}{2} (x, \delta x),\; \Theta'=\frac{1}{2}(x,\rho(\psi) x)
    \end{multline}
    where $\langle,\rangle$ denotes the canonical pairing between a vector space and its dual.
    It is a straightforward check that this comprises the data of a trivial Hamiltonian $Q$-bundle of degree $4k+2$, according to definition \ref{def: trivial Ham Q-bundles}.

    As a variation of this example, we may require that the pairing $(,)$ is anti-symmetric (so that $R$ is a symplectic space instead of Euclidean) and that $\rho: \g\ra \mathfrak{sp}(R)\subset \End(R)$ is a symplectic representation. Then formulae (\ref{g[1] structure},\ref{module structure}) define a Hamiltonian $Q$-bundle structure of degree $4k$ on $\g[1]\times R[2k]$ (note that here we choose an even degree shift for the fiber).
\item \label{ex: Ham Q bun from cyc L_infty alg plus module}
Example above generalizes straightforwardly to the setting of $L_\infty$ algebras and $L_\infty$ modules.
Namely, let $\g$ be an $L_\infty$ algebra with operations $\{l_j: \wedge^j \g\ra\g \}_{j\geq 1}$. Also, let $R$ be graded vector space with graded symmetric pairing $(,)$ of degree $q$ and with the structure of an $L_\infty$ module over $\g$ with the module operations $\{\rho_j:\wedge^j\g\otimes R\ra R\}_{j\geq 0}$, where we additionally require that operations $(\bt,\rho_j(\cdots;\bt))$ are graded anti-symmetric w.r.t. the two entries in the module. For $\psi$ a $\g$-valued degree 1 coordinate on $\g[1]$ and $x$ a $R$-valued degree $2k+1$ coordinate on $R[2k+1]$, we construct a trivial Hamiltonian $Q$-bundle with
$\MM=\g[1]$, $\NN=R[2k+1]$ and the data are:
\be \label{L_infty Ham Q structure}
Q=\sum_{j\geq 1} \frac{1}{j!}\left\langle l_j(\psi,\ldots,\psi) , \frac{\dd}{\dd \psi}\right\rangle 
\ee
for the base and
\begin{multline} \label{L_infty module fiber Ham Q structure}
\A=\sum_{j\geq 0} \frac{1}{j!} \left\langle \rho_j(\psi,\ldots,\psi;x),\frac{\dd}{\dd x} \right\rangle,\; \omega'=\frac{1}{2}(\delta x,\delta x),\\
\alpha'=\frac{1}{2}(x,\delta x),\; \Theta'=\sum_{j\geq 0} \frac{1}{2j!}(x,\rho_j(\psi,\ldots,\psi;x))
\end{multline}
This comprises the data of a trivial Hamiltonian $Q$-bundle of degree $4k+2+q$.

\item Example (\ref{ex: Ham Q bun from cyc L_infty alg plus module}) admits the following variation. If $\h$ is an $L_\infty$ algebra with $I\subset \h$ an $L_\infty$ ideal\footnote{Recall that a subspace $I\subset \h$ is called an $L_\infty$ ideal if the $L_\infty$ operations on $\h$ take values in $I$ if at least one argument is in $I$.}, then there is a natural $Q$-bundle
    \be\pi:\h[1]\ra (\h/I)[1]\label{Q bundle from L_infty ideal}\ee
    with bundle projection being the quotient map and with the fiber $I$. If in addition  we fix a section of the quotient map $\h\ra\h/I$, i.e. a splitting of the short exact sequence
    \be I\ra\h\ra \h/I \label{h I sequence}\ee
    and if
    $I$ carries a degree $q$ cyclic\footnote{By cyclic property for the pairing on the ideal we mean that extending the pairing on $I$ to $\h$ by zero on $\h/I$, we obtain a degenerate cyclic pairing for the $L_\infty$ algebra $\h$.} inner product,
    then (\ref{Q bundle from L_infty ideal}) becomes a Hamiltonian $Q$-bundle of degree $2+q$. Explicitly, the structure on the base $\MM=(\h/I)[1]$ is again given by (\ref{L_infty Ham Q structure}) with $\{l_j\}$ the quotient $L_\infty$ operations on $\h/I$, $\psi$ the degree 1 $(\h/I)$-valued coordinate on $(\h/I)[1]$. The structure on the fiber $\NN=I[1]$ is:
    \begin{multline}
    \A=\sum_{j\geq 0,k\geq 0,j+k\geq 1} \frac{1}{j!k!}\left\langle P_I\lambda_{j+k}(\underbrace{\psi,\ldots,\psi}_j,\underbrace{x,\ldots,x}_k),\frac{\dd}{\dd x}\right\rangle,\\
    \omega'=\frac{1}{2}(\delta x,\delta x),\;
    \alpha'=\frac{1}{2}(x,\delta x),\\ \Theta'=\sum_{j\geq 0,k\geq 0,j+k\geq 1} \frac{1}{j!(k+1)!}(x,P_I\lambda_{j+k}(\underbrace{\psi,\ldots,\psi}_j,\underbrace{x,\ldots,x}_k))
    \end{multline}
    where $P_I:\h\ra I$ is the projection to the ideal fixed by the choice of splitting of (\ref{h I sequence}), $\{\lambda_j\}$ are the $L_\infty$ operations on $\h$, $x$ is the degree 1 $I$-valued coordinate on $I[1]$ and $(,)$ is the cyclic pairing on $I$.

\item \label{Ham Q-bun ex Ham action} Let $\g$ be a Lie algebra and suppose there is a Hamiltonian action of $\g$ on an exact symplectic manifold $(M,\omega_M=\delta\alpha_M)$ with equivariant moment map $\mu: M\ra \g^*$. Then we set $\MM=\g[1]$ as in (\ref{ex: quadr Lie algebra with a orthog module}), with the structure (\ref{g[1] structure}). For the fiber, we set
    $\NN=M$ and
    \be \A=\{\langle\psi,\mu\rangle,\bt\},\; \omega'=\omega_M,\; \alpha'=\alpha_M,\; \Theta'= \langle\psi, \mu\rangle \label{fiber data from Ham action of g}\ee
    This is the data of a trivial Hamiltonian $Q$-bundle of degree $0$ on $\g[1]\times M$. Exactness condition for $\omega_M$ can be relaxed to the Bohr-Sommerfeld integrality condition, cf. remark \ref{rem: B-S integrality}.
\item \label{Ham Q-bun ex point} For any $Q$-manifold $(\MM,Q)$, given a function $\theta\in C^\infty(\MM)$ satisfying
$$Q\theta=0,\qquad |\theta|=p$$
we can construct a trivial Hamiltonian $Q$-bundle over $\MM$ of degree $p-1$ with fiber $\NN$ a point, $\A=0$, $\omega'=0$, $\Theta'=\theta$ (and we formally prescribe degree $p-1$ to $\omega'$).
\end{enumerate}

\section{BV pre-observables and observables}\label{sec: pre-obs and obs}
In this section we recall some basic structures of BV formalism: classical and quantum BV theory (
cf. \cite{Schwarz}), observables in classical and quantum setting .
We also introduce the notion of a \textit{pre-observable}\footnote{This is not a part of the standard BV lore; the term is sometimes used in a totally different sense in the literature.}
: an extension of a BV theory by auxiliary fields endowed with an odd-symplectic structure and an action (which is also dependent on the fields of the ambient theory). Given a pre-observable, one can construct an observable by integrating out the auxiliary fields (propositions \ref{prop: pre-obs to obs}, \ref{prop: quantum pre-obs to obs}). Pre-observables come in three versions: purely classical, \textit{semi-quantum} -- suited for integrating out auxiliary fields to produce a classical observable for the ambient theory, and \textit{quantum} -- suited for integrating out auxiliary fields to produce a quantum observable for the ambient theory.

Throughout this section in our treatment of quantum objects we are assuming that the spaces of fields are finite-dimensional. This is almost never the case in the interesting examples. In the setting of local quantum field theory, when the spaces of fields are infinite-dimensional, measures on fields are problematic to define; however the BV Laplacians and the integrals over auxiliary fields (which are perturbative path integrals now) can be made sense of with an appropriate regularization/renormalization procedure (one systematic approach is via Wilson's renormalization flow, cf. \cite{Co_renorm}). Proofs of propositions \ref{prop: pre-obs to obs}, \ref{prop: quantum pre-obs to obs} cannot be taken for granted in the infinite-dimensional setting, and should be redone within the framework of perturbative path integrals (cf. e.g. the proof of proposition \ref{prop 6} in section \ref{sec: torsion observables}).

In view of this, our treatment of quantum objects in this section should be viewed as a simplified discussion, to motivate the interest in pre-observables and argue that given a pre-observable one can expect a path integral expression for an observable.

\subsection{Pre-observables and observables for a classical ambient BV theory}
\begin{definition} \label{def: class BV theory}
We will call a \textit{(classical) BV theory} a Hamiltonian $Q$-manifold of degree $-1$, i.e. a quadruple $(\FF,Q,\Omega,S)$ consisting of a $Q$-manifold $(\FF,Q)$ (the \textit{space of fields} with the \textit{BRST operator}), the degree $-1$ symplectic structure (the \textit{BV 2-form}) $\Omega$ and the degree $0$ action $S$ satisfying
$\{S,\bt\}_\Omega=Q$ and the \textit{classical master equation}:
\be \{S,S\}_\Omega=0 \label{CME}\ee
\end{definition}

\begin{definition}\label{def: class observable}
A \textit{(classical) observable} for a classical BV theory $(\FF,Q,\Omega,S)$ is a degree $0$ function\footnote{We will be assuming degree $0$ condition for observables by default, but sometimes it is interesting to relax it. In such cases we will indicate the degree explicitly.} $O\in C^\infty(\FF)$ satisfying
$$Q(O)=0$$
Observables $O$ and $O'$ are called \textit{equivalent} if $O'-O=Q(\Psi)$ for some $\Psi\in C^\infty(\FF)$, i.e. $O'$ and $O$ give the same class in $Q$-cohomology.
\end{definition}

\begin{definition}\label{def: class pre-observable}
For a classical BV theory $(\FF,Q,\Omega,S)$ (which we will call the \textit{ambient theory}), a \textit{pre-observable} is a degree $-1$ Hamiltonian $Q$-bundle over $\FF$. We will call the fiber the \textit{space of auxiliary fields} $\FF^\aux$, which comes with its own degree $-1$ symplectic structure $\Omega^\aux$ and the \textit{action for auxiliary fields} $S^\aux\in C^\infty(\FF\times \FF^\aux)$ satisfying
\be Q S^\aux+\frac{1}{2}\{S^\aux,S^\aux\}_{\Omega^\aux}=0 \label{pre-observable CME}\ee
\end{definition}

\begin{rem}\label{rem: pre-obs to CME for S+S^aux}
Note that if in addition to (\ref{pre-observable CME}) we have the property
$$\{S^\aux,S^\aux\}_{\Omega}=0$$
then the pre-observable gives a new classical BV theory
$$(\FF\times \FF^\aux,Q+\{S^\aux,\bt\}_{\Omega+\Omega^\aux},\Omega+\Omega^\aux,S+S^\aux)$$
Note also that this cohomological vector field on $\FF\times\FF^\aux$ differs from the one arising from the $Q$-bundle structure by the term $\{S^\aux,\bt\}_\Omega$ and is in general not projectable to $\FF$.
\end{rem}

In the case when $\FF^\aux$ is finite-dimensional, it makes sense to introduce the following notion.
\begin{definition} \label{def: semi-quantum pre-observable}
For a classical BV-theory $(\FF,Q,\Omega,S)$, a \textit{semi-quantum} pre-observable is a quadruple $(\FF^\aux,\Omega^\aux,\mu^\aux,S^\aux)$ consisting of a $\ZZ$-graded manifold $\FF^\aux$ and a degree $-1$ symplectic form $\Omega^\aux$ on it. Further,
\begin{itemize}
\item $\FF^\aux$ is endowed with a volume element $\mu^\aux$ compatible with $\Omega^\aux$ in the sense that the associated BV Laplacian on $C^\infty(\FF^\aux)$,
$$\Delta^\aux: f\mapsto \frac{1}{2}\;\mr{div}_{\mu^\aux}\{f,\bt\}_{\Omega^\aux}$$
satisfies
$$(\Delta^\aux)^2=0$$
\item Degree $0$ action $S^\aux\in C^\infty(\FF\times \FF^\aux)$ satisfies
\be Q S^\aux+\frac{1}{2}\{S^\aux,S^\aux\}_{\Omega^\aux}-i\Delta^\aux S^\aux=0 \label{pre-observable QME}\ee
or, equivalently,
\be (Q-i\Delta^\aux)\;e^{i S^\aux}=0 \label{pre-observable QME exp form}\ee
\end{itemize}

We call two semi-quantum pre-observables $(\FF^\aux,\Omega^\aux,\mu^\aux,S^\aux)$ and $(\FF^\aux,\Omega^\aux,\mu^\aux,\tilde S^\aux)$ \textit{equivalent}, if there exists a degree $-1$ function $R^\aux\in C^\infty(\FF\times\FF^\aux)$ such that
\be e^{i \tilde S^\aux}-e^{i S^\aux}=(Q-i\Delta^\aux) \left(e^{i S^\aux} R^\aux\right) \label{equivalence of sq pre-obs, exp form}\ee
Infinitesimally, this equivalence condition can be written as
\be \tilde S^\aux-S^\aux=Q R^\aux + \{S^\aux,R^\aux\}_{\Omega^\aux}-i \Delta^\aux R^\aux+\mathcal{O}(R^2) \label{equivalence of sq pre-obs, infinitesimal}\ee
\end{definition}

In particular, a classical pre-observable can be promoted to a semi-quantum pre-observable if there exists a volume element $\mu^\aux$ on $\FF^\aux$ compatible with $\Omega^\aux$, such that in addition to (\ref{pre-observable CME}) we have $\Delta^\aux S^\aux=0$.

Given a semi-quantum pre-observable, one can construct an observable for the ambient classical BV theory using fiber BV integrals \cite{DiscrBF} as follows.


\begin{prop}\label{prop: pre-obs to obs} Given a semi-quantum pre-observable $(\FF^\aux,\Omega^\aux,\mu^\aux,S^\aux)$, set
\be O_\LL=\int_{\LL\subset \FF^\aux} e^{i S^\aux} \sqrt{\mu^\aux}|_\LL \in C^\infty(\FF) \label{obs via fiber BV integral}\ee
for $\LL\subset \FF^\aux$ a Lagrangian submanifold. Then:
\begin{enumerate}[(i)]
\item \label{prop 1 item 1} $O_\LL$ is an observable, i.e. $Q(O_\LL)=0$.
\item \label{prop 1 item 2} If Lagrangian submanifolds $\LL$ and $\LL'$ can be connected by a Lagrangian homotopy, then the difference $O_{\LL'}-O_{\LL}$ is $Q$-exact, i.e. observables $O_\LL$ and $O_{\LL'}$ are equivalent.
\item \label{prop 1 item 3} Given two equivalent semi-quantum pre-observables $S^\aux$ and $\tilde S^\aux$, the associated observables $O_\LL$ and $\tilde O_\LL$ are equivalent.
\end{enumerate}
\end{prop}

\begin{proof} Applying $Q$ to the integral (\ref{obs via fiber BV integral}), we obtain
\begin{equation}\label{prop 1 eq1}
Q(O_\LL)=Q\int_{\LL}  e^{i S^\aux}
=\int_{\LL} (Q-i\Delta^\aux) e^{i S^\aux}=0
\end{equation}
where we used (\ref{pre-observable QME exp form}) and the Stokes' theorem for BV integrals \cite{Schwarz}, that the BV integral of a $\Delta$-coboundary vanishes. We suppress the measure in the short-hand notation for integrals here. This proves (\ref{prop 1 item 1}).

Now let $\{\LL_t\}_{t\in [0,1]}$ be a smooth family of Lagrangian submanifolds of $(\FF^\aux,\Omega^\aux)$ connecting $\LL=\LL_0$ and $\LL'=\LL_1$. An infinitesimal deformation of a Lagrangian submanifold $\LL_t$ is given by a graph of 
a Hamiltonian vector field\footnote{A general infinitesimal Lagrangian deformation is given by $\mr{graph}(\iota_{(\Omega^{\aux})^{-1}}\chi_t)$ for any degree $-1$ \textit{closed} 1-form $\chi_t\in\Omega^1(\LL_t)$. Since in non-zero degree a closed form is automatically \textit{exact} (cf. remark \ref{rem: symp form is exact}), we have (\ref{inf transf of L}) with the generator $\Psi_t=\iota_E \chi_t$.}
\be \mbox{``} \frac{d}{dt}\LL_t \mbox{''} =\mr{graph}(\{\bt,\Psi_t\}_{\Omega^\aux})\in \Gamma(\LL_t,N\LL_t) \label{inf transf of L}\ee
viewed as a section of the normal bundle of $\LL_t$,
with $\Psi_t\in C^\infty(\LL_t)$, $|\Psi_t|=-1$ the generator. Hence,
\begin{multline*}
\frac{d}{dt}\; O_{\LL_t}= \frac{d}{dt}\; \int_{\LL_t} e^{i S^\aux|_{\LL_t}} \\
=\int_{\LL_t} \{e^{i S^\aux},\Psi_t\}_{\Omega^\aux} +
e^{i S^\aux}\cdot \Delta^\aux(\Psi_t) \\
=\int_{\LL_t} \Delta^\aux \left(e^{i S^\aux}\Psi_t\right) -\Delta^\aux \left(e^{i S^\aux}\right) \cdot \Psi_t \\
= \int_{\LL_t} \Delta^\aux \left(e^{i S^\aux}\Psi_t\right) +i Q \left(e^{i S^\aux}\right) \cdot \Psi_t \\
= Q \left(i \int_{\LL_t} e^{i S^\aux} \Psi_t \right)
\end{multline*}
The term $\Delta^\aux(\Psi_t)$ in the second line comes from the transformation of measure $\sqrt{\mu^\aux}|_{\LL_t}$ by the infinitesimal shift of Lagrangian (\ref{inf transf of L}). This calculation implies that
$$ O_{\LL'}-O_\LL = Q\left(\int_0^1 dt\; i \int_{\LL_t} e^{i S^\aux} \Psi_t  \right)$$
and thus proves (\ref{prop 1 item 2}).

Item (\ref{prop 1 item 3}) follows immediately from the construction (\ref{obs via fiber BV integral}), our definition of equivalence (\ref{equivalence of sq pre-obs, exp form}) and the Stokes' theorem for BV integrals:
\begin{multline*}
\tilde O_\LL - O_\LL=\int_\LL e^{i \tilde S^\aux} - e^{i S^\aux}\\
=\int_\LL (Q-i \Delta^\aux)\left(e^{i S^\aux}R^\aux\right)= Q\left(\int_\LL e^{i S^\aux}R^\aux\right)
\end{multline*}

\end{proof}
Here we are implicitly assuming convergence of the integral (\ref{obs via fiber BV integral}), and for the proof of (\ref{prop 1 item 2}) the integral over $\LL_t$ should converge for every $t\in [0,1]$.

The following construction allows one to produce new semi-quantum pre-observables out of old ones by partially integrating out the auxiliary fields.
\begin{prop}
\label{prop: pushing pre-obs to zero-modes}
Given a semi-quantum pre-observable $(\FF^\aux,\Omega^\aux,\mu^\aux,S^\aux)$, assume that $\FF^\aux$ is given as a product\footnote{Subscript ``z'' refers to auxiliary zero-modes.}, $\FF^\aux=\FF^\aux_\z\times \tilde\FF^\aux$, so that $\Omega^\aux$ and $\mu^\aux$ split:
$$\Omega^\aux=\Omega^\aux_\z+\tilde\Omega^\aux,\quad \mu^\aux=\mu^\aux_\z\times \tilde\mu^\aux$$
Define $S^\aux_\z\in C^\infty(\FF\times \FF^\aux_\z)$ by
$$e^{i S^\aux_\z}=\int_{\tilde\LL\subset\tilde\FF^\aux} e^{i S^\aux} \sqrt{\tilde\mu^\aux}|_{\LL}$$
with $\tilde\LL$ a Lagrangian submanifold of $(\tilde\FF^\aux,\tilde\Omega^\aux)$.
Then
\begin{enumerate}[(i)]
\item $(\FF^\aux_\z, \Omega^\aux_\z,\mu^\aux_\z, S^\aux_\z)$ is a semi-quantum pre-observable for the same ambient classical BV theory.
\label{prop 1.5 item 1}
\item The observable for the ambient theory induced from $S^\aux_z$ using (\ref{obs via fiber BV integral}) with Lagrangian $\LL_\z\subset \FF^\aux_\z$ is equivalent to the one induced directly from $S^\aux$ using Lagrangian $\LL\subset \FF^\aux$, provided that $\LL$ can be connected with $\LL_\z\times\tilde\LL$ by a Lagrangian homotopy in $\FF^\aux$.
    \label{prop 1.5 item 2}
\end{enumerate}
\end{prop}

\begin{proof} Similarly to (\ref{prop 1 eq1}), we have
$$ (Q-i \Delta^\aux_\z)\; e^{i S^\aux_\z}= \int_{\tilde\LL} (Q\underbrace{-i \Delta^\aux_\z-i \tilde\Delta^\aux}_{-i\Delta^\aux})\; e^{i S^\aux}= 0 $$
which proves (\ref{prop 1.5 item 1}). Item (\ref{prop 1.5 item 2}) is an immediate corollary of (\ref{prop 1 item 2}) of proposition \ref{prop: pre-obs to obs}.
\end{proof}

\subsection{Pre-observables and observables for a quantum ambient BV theory}
In the case when the space of fields of the ambient theory $\FF$ is finite-dimensional, it makes sense to introduce the following notions.
\begin{definition}[cf. \cite{Schwarz}] \label{def: quantum BV theory}
A \textit{quantum BV theory} is a quadruple $(\FF,\Omega,\mu,S)$ where $(\FF,\Omega)$ is a $\ZZ$-graded manifold with a degree $-1$ symplectic form $\Omega$;
$\mu$ is a volume element on $\FF$, compatible with $\Omega$ as in definition \ref{def: semi-quantum pre-observable}, i.e. the BV Laplacian $\Delta: f\mapsto \frac{1}{2}\mr{div}_\mu\{f,\bt\}_\Omega$ satisfies $\Delta^2=0$; the degree $0$ action $S\in C^\infty(\FF)$
is required to satisfy the \textit{quantum master equation}
\be \frac{1}{2}\{S,S\}_\Omega-i\Delta S=0 \label{QME}\ee
or, equivalently,
$$\Delta\, e^{i S}=0$$
\end{definition}
For a quantum BV theory it is convenient to introduce a degree $1$ second order differential operator
$$\delta_\BV=\{S,\bt\}_\Omega-i \Delta = e^{-i S}(-i \Delta) e^{i S} $$
which satisfies $(\delta_\BV)^2=0$.

\begin{definition} \label{def: quantum observable}
A \textit{quantum observable}\footnote{Here we define a quantum observable in the framework of Lagrangian field theory in BV formalism, as something that can be averaged over the space of fields to yield a correlator (\ref{correlator}). In other contexts one has quite different definitions of quantum observables: e.g. in the framework of Atiyah's topological quantum field theory \cite{LosevTQFT}, an observable is a pair consisting of a submanifold $\gamma\subset\Sigma$ of the spacetime manifold and a vector $\hat O_\gamma$ in the space of states associated to the boundary of a tubular neighborhood of $\gamma$ in $\Sigma$. The passage from the first picture to the second one goes through path integral quantization of the ambient topological theory on the tubular neighborhood (as on a manifold with boundary) with the insertion of observable.}
for a quantum BV theory is a degree $0$ function $O\in C^\infty(\FF)$ satisfying
$$\delta_\BV O=0$$
which is equivalent to
$$\Delta\left(O\, e^{i S}\right)=0$$
Observables $O$ and $O'$ are called \textit{equivalent} if $O'-O=\delta_\BV (\Psi)$ for some $\Psi\in C^\infty(\FF)$, i.e. $O'$ and $O$ give the same class in $\delta_\BV$-cohomology.

\end{definition}

Given a quantum observable, one can define its \textit{correlator} as
\be \langle O\rangle_\LL= \frac{\int_{\LL\subset \FF} O\; e^{i S} \sqrt{\mu}|_\LL}{\int_{\LL\subset \FF} e^{i S} \sqrt{\mu}|_\LL}\in \mathbb{C} \label{correlator}\ee
where $\LL$ is a Lagrangian submanifold of $(\FF,\Omega)$ and we are assuming convergence. By the Stokes' theorem for BV integrals, this expression does not change with Lagrangian homotopy of $\LL$: $\langle O\rangle_{\LL}=\langle O\rangle_{\LL'}$, and also correlators of equivalent observables coincide: $\langle O\rangle_{\LL}=\langle O'\rangle_{\LL}$, cf. \cite{Schwarz}.

Finally, if both $\FF$ and $\FF^\aux$ are finite-dimensional, the following definition makes sense.
\begin{definition} \label{def: quantum pre-observable}
A \textit{quantum pre-observable} for a quantum BV theory is a semi-quantum pre-observable $(\FF^\aux, \Omega^\aux,\mu^\aux, S^\aux)$  where instead of equation (\ref{pre-observable QME}) we require that $S+S^\aux$ satisfies the quantum master equation on $\FF\times \FF^\aux$:
$$(\Delta+\Delta^\aux)\,e^{i(S+S^\aux)}=0$$
\end{definition}

The analog of proposition \ref{prop: pre-obs to obs} in the setting of quantum BV theory is as follows.
\begin{prop}  \label{prop: quantum pre-obs to obs}
Given a quantum pre-observable $(\FF^\aux,\Omega^\aux,\mu^\aux,S^\aux)$,
define again $O_\LL$ 
as  $$O_\LL=\int_{\LL\subset \FF^\aux} e^{i S^\aux} \sqrt{\mu^\aux}|_\LL \in C^\infty(\FF)$$
for $\LL\subset \FF^\aux$ a Lagrangian submanifold.
Then:
\begin{enumerate}[(i)]
\item $O_\LL$ is a quantum observable, i.e. $\delta_\BV O_\LL=0$.
\item If Lagrangian submanifolds $\LL$ and $\LL'$ can be connected by a Lagrangian homotopy, then the difference $O_{\LL'}-O_{\LL}$ is $\delta_\BV$-exact.
\end{enumerate}
\end{prop}

\begin{proof}
The proof is obtained from the proof of proposition \ref{prop: pre-obs to obs} by replacing $Q$ by $-i\Delta$, $S^\aux$ by $S+S^\aux$ and $O$ by $O\, e^{i S}$ everywhere.
\end{proof}

Similarly, proposition \ref{prop: pushing pre-obs to zero-modes} holds in the context of quantum pre-observables.


\begin{rem}\label{rem: Planck's constant} It is customary for perturbative quantization to introduce a formal parameter (the ``Planck's constant'') $\hbar$ by making a rescaling
\be \Omega\ra \hbar^{-1}\Omega,\quad S\ra \hbar^{-1} S \label{Omega and S rescaling}\ee
(which means $\Omega^\mr{old}=\hbar^{-1}\Omega^\mr{new}$ etc.) where the new action is a formal power series in $\hbar$, $S\in C^\infty(\FF)[[\hbar]]$. Hamiltonian vector field $\{S,\bt\}_\Omega$ does not rescale and the BV Laplacian rescales as $\Delta\ra \hbar\Delta$. For the structure on auxiliary fields, one does the same:
\be \Omega^\aux\ra\hbar^{-1}\Omega^\aux,\quad S^\aux\ra \hbar^{-1}S^\aux \label{Omega and S aux rescaling}\ee
With these redefinitions the equation (\ref{pre-observable QME}) on semi-quantum pre-observables becomes
$$Q S^\aux+\frac{1}{2}\{S^\aux,S^\aux\}_{\Omega^\aux}-i\hbar \Delta^\aux S^\aux=0$$
and can be solved by obstruction theory order by order in $\hbar$ for $S^\aux=S^\aux_{(0)}+\hbar\, S^\aux_{(1)}+\hbar^2\, S^\aux_{(2)}+\cdots$ starting from $S^\aux_{(0)}$ a solution of (\ref{pre-observable CME}). Likewise, the BV push-forward formula (\ref{obs via fiber BV integral}) becomes
$$ O=\int_{\LL\subset \FF^\aux} e^{\frac{i}{\hbar}S^\aux} \sqrt{\mu^\aux}|_\LL  \in C^\infty(\FF)[[\hbar]]$$
and can be evaluated by the stationary phase formula.

Note that one can also introduce two independent Planck's constants $\hbar$, $\hbar^\aux$ for the ambient (\ref{Omega and S rescaling})  and auxiliary (\ref{Omega and S aux rescaling}) theory respectively.
\end{rem}

\section{AKSZ pre-observable associated to a Hamiltonian $Q$-bundle}
\label{sec: AKSZ pre-observable}
Let $(\MM,Q,\omega=\delta\alpha,\Theta)$ be a Hamiltonian $Q$-manifold and $\Sigma$ a closed oriented spacetime manifold, $\dim\Sigma=|\omega|+1$. Then we have the AKSZ theory with \begin{multline}\label{AKSZ theory}
\FF_\Sigma=\Map(T[1]\Sigma,\MM),\qquad Q_\Sigma=(d_\Sigma)^{\lifted}+Q^{\lifted}, \\ \Omega_\Sigma=(-1)^{\dim\Sigma}\tau_\Sigma(\omega),\qquad S_\Sigma=\iota_{d_\Sigma^{\lifted}}\tau_\Sigma(\alpha)+\tau_\Sigma(\Theta)
\end{multline}
as in section \ref{sec: AKSZ reminder}.

Given a trivial Hamiltonian $Q$-bundle over $\MM$ with fiber data $(\NN,\A,\omega'=\delta\alpha',\Theta')$, a closed oriented manifold $\gamma$ of dimension $\dim\gamma=|\omega'|+1$ and a smooth map\footnote{By default, we assume that $\ii$ is an embedding, cf. remark \ref{rem: i embedding} below.} $\ii:\gamma\ra \Sigma$, we can construct a pre-observable for the AKSZ theory (\ref{AKSZ theory}) with target $\MM$ as follows. Set
\begin{multline}\label{AKSZ pre-observable}
\FF_\gamma=\Map(T[1]\gamma,\NN),\qquad \A_\gamma=(d_\gamma)^\lifted+ p^*\A^\lifted,\\
\Omega_\gamma=(-1)^{\dim\gamma}\tau_\gamma(\omega'),\qquad S_\gamma= \iota_{(d_\gamma)^\lifted}\tau_\gamma(\alpha')+p^*\tau_\gamma^\tot(\Theta')
\end{multline}
Here  the notations are as follows.
\begin{itemize}
\item $\tau_\gamma: \Omega^\bt(\NN)\ra \Omega^\bt(\FF_\gamma)$ and $\tau_\gamma^\tot:\Omega^\bt(\EE)\ra \Omega^\bt(\Map(T[1]\gamma,\EE))$ are the transgression maps, defined as in (\ref{AKSZ transgression}); $\EE=\MM\times\NN$ is the total space of the target Hamiltonian $Q$-bundle.
\item Map $p=\ii^*: \FF_\Sigma\ra \Map(T[1]\gamma,\MM)$ is the pull-back of ambient fields by
$\ii:\gamma\ra\Sigma$
and $p^*: C^\infty(\Map(T[1]\gamma,\MM))\ra C^\infty(\FF_\Sigma)$ is the pull-back by $p$. By the trivial extension to auxiliary fields, $p=p\times \id_{\FF_\gamma}$ also maps $\FF_\Sigma\times \FF_\gamma$ to $\Map(T[1]\gamma,\EE)$.
\item
$(d_\gamma)^\lifted\in \Vect(\FF_\gamma)\subset \Vect^\ver(\FF_\Sigma\times \FF_\gamma)$ is the lifting of the de Rham differential on $\gamma$ to a vector field on the mapping space $\Map(T[1]\gamma,\NN)$.
\item $\A^\lifted$ is the lifting of the vertical vector field $\A\in\Vect^\ver(\EE)$ to a vertical vector field on the mapping space $\Map(T[1]\gamma,\EE)$, so that $p^*\A^\lifted$ becomes a vertical vector field on $\FF_\Sigma\times \FF_\gamma$.
\end{itemize}

\begin{prop} \label{prop 3}
The data (\ref{AKSZ pre-observable}) defines a classical pre-observable for the AKSZ theory (\ref{AKSZ theory}), in particular
\be Q_\Sigma S_\gamma + \frac{1}{2}\{S_\gamma,S_\gamma\}_{\Omega_\gamma}=0 \label{prop 3 CME}\ee
\end{prop}

\begin{proof}
First let us check that $\A_\gamma$ as defined in (\ref{AKSZ pre-observable}) is the Hamiltonian vector field for $S_\gamma$. Introduce the notations for the source (kinetic) and target parts of $\A_\gamma$ and $S_\gamma$:
$$\A_\gamma=\underbrace{(d_\gamma)^\lifted}_{\A_\gamma^\kin}+ \underbrace{p^*\A^\lifted}_{\A_\gamma^\target},\qquad S_\gamma= \underbrace{\iota_{(d_\gamma)^\lifted}\tau_\gamma(\alpha')}_{S_\gamma^\kin}+\underbrace{p^*\tau_\gamma^\tot(\Theta')}_{S_\gamma^\target}$$
Then we have
\begin{multline} \label{prop 3 eq1}
\iota_{\A^\kin_\gamma}\Omega_\gamma = \iota_{(d_\gamma)^\lifted} (-1)^{\dim\gamma}\tau_\gamma(\omega')= \iota_{(d_\gamma)^\lifted}\delta\tau_\gamma(\alpha')= \\
=\underbrace{L_{(d_\gamma)^\lifted}\tau_\gamma(\alpha')}_{=0\;\mathrm{by}\; (\ref{L_v tau = 0})}+
\delta \iota_{(d_\gamma)^\lifted} \tau_\gamma(\alpha')=\delta S^\kin_\gamma=\delta^\ver S^\kin_\gamma
\end{multline}
and
\begin{multline} \label{prop 3 eq2}
\iota_{\A^\target_\gamma}\Omega_\gamma=p^*\iota_{\A^\lifted}(-1)^{\dim\gamma}\tau^\tot_\gamma(\omega')=\\
=(-1)^{\dim\gamma}p^*\tau^\tot_\gamma(\underbrace{\iota_\A\omega'}_{=\delta^\ver \Theta'})=\delta^\ver p^*\tau^\tot_\gamma(\Theta')=\delta^\ver S_\gamma^\target
\end{multline}
Here $\delta^\ver$ stands for the de Rham differential in the fiber direction, 
as in section \ref{sec: general Ham Q-bundles}.
Collecting (\ref{prop 3 eq1},\ref{prop 3 eq2}), we get
$$\{S_\gamma,\bt\}_{\Omega_\gamma}=\A_\gamma$$

Next, let us prove (\ref{prop 3 CME}). Note that
$$Q_\Sigma S_\gamma^\kin=0$$
since $S_\gamma^\kin$ does not depend on the ambient fields (i.e. is a function on $\FF_\Sigma\times\FF_\gamma$, constant in the direction of $\FF_\Sigma$); also
$$Q_\Sigma^\kin S_\gamma^\target=p^*(Q_\Sigma^\kin|_\gamma \tau^\tot_\gamma(\Theta'))=0$$
as in (\ref{L_v tau = 0}), where $Q_\Sigma^\kin|_\gamma$ is the lifting of $d_\gamma$ to $\Map(T[1]\gamma,\MM)$, extended to a horizontal vector field on $\Map(T[1]\gamma,\EE)$.
Thus
\begin{multline} \label{prop 3 CME calculation}
Q_\Sigma S_\gamma + \frac{1}{2}\{S_\gamma,S_\gamma\}_{\Omega_\gamma}=\\
=Q_\Sigma^\target S_\gamma^\target +\frac{1}{2}\{S_\gamma^\kin,S_\gamma^\kin\}_{\Omega_\gamma}+\{S_\gamma^\kin,S_\gamma^\target\}+\frac{1}{2} \{S_\gamma^\target,S_\gamma^\target\}_{\Omega_\gamma} \\
=\frac{1}{2}\{S_\gamma^\kin,S_\gamma^\kin\}_{\Omega_\gamma} + \underbrace{\A_\gamma^\kin S_\gamma^\target}_{0\;\mathrm{by}\;(\ref{L_v tau = 0})} + (Q_\Sigma^\target+\frac{1}{2}\A_\gamma^\target)S_\gamma^\target
\end{multline}
Here the first term on the r.h.s. has $((d_\gamma)^\lifted)^2=0$ as its Hamiltonian vector field on $\FF_\gamma$, thus it is a degree 1 constant function on $\FF_\gamma$, and hence vanishes. Therefore, continuing the calculation (\ref{prop 3 CME calculation}) and using (\ref{L_u tau = tau L_u}), we have
\begin{multline}
Q_\Sigma S_\gamma + \frac{1}{2}\{S_\gamma,S_\gamma\}_{\Omega_\gamma}=
(-1)^{\dim\gamma} p^*\tau^\tot_\gamma(Q\Theta'+\frac{1}{2}\A\Theta')=\\
=(-1)^{\dim\gamma} p^* \tau^\tot_\gamma(Q\Theta'+\frac{1}{2}\{\Theta',\Theta'\}_{\omega'}) = 0
\end{multline}
since the target is a Hamiltonian $Q$-bundle and equation (\ref{fiber CME}) holds there.
This finishes the proof.
\end{proof}

\textit{In coordinates.}
Alongside with the superfield (\ref{superfield}) for the ambient theory on $\Sigma$, we now have a superfield for the auxiliary theory on $\gamma$:
\begin{multline}\label{aux superfield}
Y^i(v^1,\ldots,v^{\dim\gamma},\xi^1,\ldots,\xi^{\dim\gamma})= \\
=\sum_{k=1}^{\dim\gamma}\;\underbrace{\sum_{1\leq \nu_1<\cdots < \nu_k\leq \dim\gamma} Y^i_{\nu_1\cdots\nu_k}(v^1,\ldots,v^{\dim\gamma})\;\xi^{\nu_1}\cdots\xi^{\nu_k}}_{Y^i_{(k)}(v,\xi)}
\end{multline}
associated to local homogeneous coordinates $y^i$ on $\NN$.
Here $v^1,\ldots,v^{\dim\gamma}$ are local coordinates on $\gamma$ and
$\xi^\nu$ stands for $dv^\nu$, $\nu=1,\ldots,\dim\gamma$.

Let $\alpha'$ and $\omega'$ be locally given on $\NN$ as $\alpha'=\alpha'_i(y)\, \delta y^i$, $\omega'=\frac{1}{2} \omega'_{ij}(y)\, \delta y^i \wedge \delta y^j$.
Then the BV 2-form for auxiliary fields is:
$$\Omega_\gamma=(-1)^{\dim\gamma}\int_\gamma\frac{1}{2}\, \omega'_{ij}(Y)\, \delta Y^i\wedge \delta Y^j$$
Its primitive is:
$$\alpha_\gamma= \int_\gamma \alpha'_i(Y)\, \delta Y^i$$
The action for auxiliary fields is:
$$S_\gamma=\int_\gamma \alpha'_i(Y)\, d Y^i+ \int_\gamma\Theta'(\ii^*X,Y)$$
Vertical part of the cohomological vector field $\A_\gamma$ acts on the auxiliary superfield by
$$\A_\gamma Y^i= d Y^i+ \A^i(\ii^*X,Y)$$
if on the target we have $\A=\A^i(x,y)\frac{\dd}{\dd y^i}$.

\begin{rem}\label{rem: S_Sigma+S_gamma CME} In the case when $i(\gamma)\subset\Sigma$ is an embedded submanifold of positive codimension, the expression
\be \{S_\gamma,S_\gamma\}_{\Omega_\Sigma} \label{(S_gamma,S_gamma) wrt Omega_Sigma}\ee
is ill-defined, because it formally contains $\int_\Sigma (\delta_\gamma)^2\cdots$ where $\delta_\gamma$ is the delta function supported on on $i(\gamma)\subset\Sigma$. On the other hand, if we use adapted local coordinates $u^1,\ldots,u^{\dim\Sigma}$ on $\Sigma$ in which $i(\gamma)$ is given by $u^{\dim\gamma+1}=\cdots=u^{\dim\Sigma}=0$, the auxiliary action $S_\gamma$ only depends on the components of the ambient superfield $X^a_{\mu_1\cdots\mu_k}$ with $\mu_1,\ldots,\mu_k\leq \dim\gamma$, whereas the BV 2-form $\Omega_\Sigma$ couples the components $X^a_{\mu_1\cdots\mu_k}$ and $X^b_{\nu_1\cdots\nu_l}$ if and only if $\{\mu_1,\ldots,\mu_k\}\sqcup\{\nu_1,\ldots,\nu_l\}=\{1,\ldots,\dim\Sigma\}$ (for the sake of this argument we choose local Darboux coordinates on $\MM$ in which $\omega_{ab}(x)$ does not depend on $x$). So, formally (\ref{(S_gamma,S_gamma) wrt Omega_Sigma}) is $\infty\cdot 0$, and we naively regularize it to zero. Thus, by remark \ref{rem: pre-obs to CME for S+S^aux}, we have (in the sense of our regularization) the classical master equation on $\FF_\Sigma\times\FF_\gamma$:
\be \{S_\Sigma+S_\gamma,S_\Sigma+S_\gamma\}_{\Omega_\Sigma+\Omega_\gamma}=0 \label{S_Sigma+S_gamma CME}\ee

In the case when $\ii:\gamma\ra\Sigma$ is a diffeomorphism, we have 
$$ \{S_\gamma,S_\gamma\}_{\Omega_\Sigma}=  \{S_\gamma^\target,S_\gamma^\target\}_{\Omega_\Sigma}=(-1)^{\dim\Sigma} p^*\tau_\gamma^\tot(\{\Theta',\Theta'\}_\omega)$$
Thus, (\ref{S_Sigma+S_gamma CME})  holds if and only if on the target we have
\be \{\Theta',\Theta'\}_\omega=0 \label{(Theta',Theta')_omega=0}\ee
in addition to (\ref{fiber CME}).
\end{rem}

\begin{rem} \label{rem: S_kin for omega' integral}
If $\dim\gamma=1$, $\omega'$ is not automatically exact. The interesting case is when $\omega'$ is not exact but satisfies Bohr-Sommerfeld integrality condition, see remark \ref{rem: B-S integrality}. We assume that $\gamma$ is a circle; the discussion extends to a disjoint union of circles trivially. The kinetic part of the action for auxiliary fields $S_\gamma^\kin$ cannot be defined as in (\ref{AKSZ pre-observable}), because the primitive $\alpha'$ does not exist globally on $\NN$. Instead one can define
\be e^{i S_\gamma^\kin(Y)} := \mr{Hol}_{Y_{(0)}^*\nabla'} (\gamma)\in U(1) \label{exp i S_gamma^kin}\ee
--  the holonomy around $\gamma$ of the pull-back of the connection $\nabla'$ (cf. remark \ref{rem: B-S integrality}) by $Y_{(0)}:\gamma\ra\NN$. Here the l.h.s. should be viewed as one symbol. Target part of the action is defined as before, so we have
$$e^{i S_\gamma(X,Y)}= e^{i\oint_\gamma \Theta'(\ii^*X,Y)}\mr{Hol}_{Y_{(0)}^*\nabla'} (\gamma)$$
Otherwise, using Stokes' theorem one can rewrite the kinetic part of the auxiliary action as
\be S_\gamma^\kin (Y) = \int_D \tilde Y_{(0)}^*\omega' 
\label{S_gamma^kin via Stokes}\ee
for $D$ a disc 
bounded by $\gamma$ and $\tilde Y_{(0)}$ an arbitrary extension of $Y_{(0)}$ from $\gamma$ to $D$. Due to integrality of $\omega'/2\pi$, arbitrariness of 
$\tilde Y_{(0)}$ results in (\ref{S_gamma^kin via Stokes}) being defined modulo multiples of $2\pi$:
$$S_\gamma^\kin (Y)\in \RR/2\pi\ZZ$$
From this point of view, the l.h.s. of (\ref{exp i S_gamma^kin}) is well-defined.

It is easy to see that the proof of proposition \ref{prop 3} still works in this case (e.g. because one can find a primitive of $\omega'$ in the tubular neighborhood of the loop $Y_{(0)}(\gamma)\subset \NN$).

\end{rem}

\begin{rem}
Pre-observable (\ref{AKSZ pre-observable}) is invariant w.r.t. $\Diff(\gamma)$ (reparametrizations) and $\Diff(\Sigma)$ (ambient diffeomorphisms) in the following sense. Let us denote explicitly the dependence of $S_\gamma$ on $\ii:\gamma\ra\Sigma$ (through $p=\ii^*$ in (\ref{AKSZ pre-observable})) as $S_\gamma(X,Y;\ii)$. Then for an ambient diffeomorphism $\Phi\in\Diff(\Sigma)$ and a reparametrization $\phi\in\Diff(\gamma)$ we have
\be S_\gamma(X,Y;\Phi\circ \ii\circ \phi) = S_\gamma(\Phi^* X,(\phi^{-1})^* Y;\ii) \label{pre-obs diffeo}\ee
which can be seen directly from the definition (\ref{AKSZ pre-observable}).
\end{rem}


\section{From pre-observables to observables}
\label{sec: from pre-obs to obs}
The general idea of passage from pre-observables constructed in section \ref{sec: AKSZ pre-observable} to observables for the underlying AKSZ theory is by means of integrating out the auxiliary fields as in proposition \ref{prop: pre-obs to obs}, i.e.
\be O_\gamma=\int_{\LL\subset \FF_\gamma} e^{i S_\gamma}\in C^\infty(\FF_\Sigma) \label{AKSZ observable via integral over aux}\ee
where $\LL$ is a Lagrangian submanifold of $(\FF_\gamma,\Omega_\gamma)$. When we want to emphasize the dependence of $O_\gamma$ on the map $\ii:\gamma\ra\Sigma$, we will write $O_{\gamma,\ii}$.

Assuming that we can make sense of the path integral (\ref{AKSZ observable via integral over aux}),
the observable $O_\gamma$ is expected to have the following properties:
\begin{enumerate}[(i)]
\item \label{obs expected i} $O_\gamma$ depends on the fields of the ambient AKSZ theory only via pull-back by $\ii:\gamma\ra\Sigma$.
\item \label{obs expected ii} Gauge invariance: $Q_\Sigma O_\gamma=0$ (which is, indeed, our definition of an observable).
\item  \label{obs expected iii} The class of $O_\gamma$ in $Q_\Sigma$-cohomology is independent of deformations of the gauge fixing $\LL$ (as a Lagrangian submanifold of $\FF_\gamma$).
\item \label{obs expected iv} The class of $O_\gamma$ in $Q_\Sigma$-cohomology is invariant under isotopy of $\gamma$ (reparametrizations of $\gamma$ homotopic to identity): for $\phi\in \mr{Diff}_0(\gamma)$, we have
    $$O_{\gamma,\ii\circ\phi}= O_{\gamma,\ii}+Q_\Sigma(\cdots)$$
\item  \label{obs expected v} Invariance under ambient diffeomorphisms: for $\Phi\in\Diff(\Sigma)$, we have
$$O_{\gamma,\Phi\circ \ii}(X)=O_{\gamma,\ii}(\Phi^*X)$$
\item \label{obs expected vi} Invariance of the correlator $\langle O_{\gamma} \rangle $ under ambient isotopy: for $\Phi\in\Diff_0(\Sigma)$, we have
$$\langle O_{\gamma,\Phi\circ \ii} \rangle = \langle O_{\gamma,\ii} \rangle$$
\end{enumerate}
Here (\ref{obs expected i}) follows from our construction of the pre-observable (\ref{AKSZ pre-observable}) and (\ref{obs expected v}) follows from (\ref{obs expected i}). Properties (\ref{obs expected ii}) and (\ref{obs expected iii}) are expected to hold in view of the proposition \ref{prop: pre-obs to obs} for finite-dimensional fiber BV integrals. Property (\ref{obs expected iv}) formally follows from (\ref{pre-obs diffeo}), (\ref{obs expected iii}) and $\Diff(\gamma)$-invariance of the measure on $\FF_\gamma$:
\begin{multline}
O_{\gamma,\ii\circ\phi}=\int_{\LL\subset \FF_\gamma} \DD Y\; e^{i S_{\gamma}(X,Y;\ii\circ\phi)}= \int_\LL \DD Y\; e^{i S_{\gamma}(X,(\phi^{-1})^*Y;\ii)}\\
=\int_{(\phi^{-1})^*\LL} (\phi^{-1*})_*(\DD Y)\; e^{i S_{\gamma}(X,Y;\ii)}
=\int_{(\phi^{-1})^*\LL} \DD Y\; e^{i S_{\gamma}(X,Y;\ii)} \\
= \int_{\LL} \DD Y\; e^{i S_{\gamma}(X,Y;\ii)}+Q_\Sigma(\cdots)
=O_{\gamma,\ii}+Q_\Sigma(\cdots)
\end{multline}
Here $\int_\LL \DD Y\cdots$ should be understood as $\int_\LL \sqrt{\mu}|_\LL\cdots$ where $\mu$ is the path integral measure on $\FF_\gamma$.
Note that we do use the fact that $\phi$ is an isotopy rather than a general diffeomorphism, since otherwise Lagrangians $\LL$ and $(\phi^{-1})^*\LL$ are not guaranteed to be homotopic.
Property (\ref{obs expected vi}) formally follows from (\ref{obs expected v}), $\Diff(\Sigma)$-invariance of the AKSZ action (\ref{AKSZ action diff-invar}) and of the path integral measure on $\FF_\Sigma$, and from the BV-Stokes' theorem:
\begin{multline}
\langle O_{\gamma,\Phi\circ \ii}\rangle = \int_{\LL\subset \FF_\Sigma} \DD X\; O_{\gamma,\Phi\circ\ii}(X)\; e^{i S_\Sigma(X)} =
\int_{\LL} \DD X\; O_{\gamma,\ii}(\Phi^*X)\; e^{i S_\Sigma(X)} \\
=\int_{\Phi^*\LL} (\Phi^*)_*(\DD X)\; O_{\gamma,\ii}(X)\; e^{i S_\Sigma((\Phi^{-1})^*X)}= \int_{\Phi^*\LL} \DD X\; O_{\gamma,\ii}(X)\; e^{i S_\Sigma(X)}\\
=\int_{\LL} \DD X\; O_{\gamma,\ii}(X)\; e^{i S_\Sigma(X)} = \langle O_{\gamma,\ii} \rangle
\end{multline}

\begin{rem}\label{rem: i embedding}
For our construction of observable $O_\gamma$ it is not a priori necessary to impose any restrictions on the map $\ii:\gamma\ra\Sigma$. However, if we want to make sense of the correlator $\langle O_\gamma \rangle$ via perturbation theory for the path integral (which is outside of the scope of this paper), we have to require that $\ii$ is an embedding. The technical reason is that if $\ii:\gamma\ra\Sigma$ is not an embedding, it does not lift to a map between compactified configuration spaces of pairs of points on $\gamma$ and $\Sigma$, $\mr{Conf}_2(\gamma)\ra \mr{Conf}_2(\Sigma)$, which prevents one from defining the pull-back of the propagator of the ambient theory to $\gamma$, which in turn leads to Feynman diagrams for the correlator being ill-defined.
\end{rem}

We will not try to give meaning to the path integral (\ref{AKSZ observable via integral over aux}) in the most general situation here, but rather will discuss several special cases.

\subsection{Case $\NN=\mathrm{point}$}
\label{sec: obs with N=pt}
In the case when $\NN$ is a point the target Hamiltonian $Q$-bundle is described in example (\ref{Ham Q-bun ex point}) in section \ref{sec: examples of Ham Q-bun}: $\A=0$, $\omega'=0$, $\Theta'=\theta\in C^\infty(\MM)$ --- a $Q$-cocycle of degree $p$ for the target $\MM$. The associated pre-observable (\ref{AKSZ pre-observable}) for a $p$-dimensional closed manifold $\gamma$ and a map $\ii:\gamma\ra\Sigma$
is:
$$\FF_\gamma=\mathrm{point},\quad \A_\gamma=0,\quad \Omega_\gamma=0,\quad S_\gamma(X)=\int_\gamma \theta(\ii^*X)$$
The passage to the observable (\ref{AKSZ observable via integral over aux}) is trivial in this setting, since there are no auxiliary fields:
\be O_\gamma=e^{i \int_\gamma \theta(\ii^*X)} \label{N=pt obs}\ee
This observable obviously satisfies properties (\ref{obs expected i}--\ref{obs expected v}) above.

\subsection{Case of 1-dimensional observables.}
\label{sec: 1D obs}
Let $\ii:\gamma=S^1 \ra \Sigma$ be 
a circle mapped
into $\Sigma$ and suppose we have a Hamiltonian $Q$-bundle of degree $0$ over the AKSZ target $(\MM,Q,\omega=\delta\alpha,\Theta)$ with fiber data $(\NN,\A,\omega'=\delta\alpha',\Theta')$. We will assume for simplicity that $\NN$ is concentrated in degree $0$, i.e. $(\NN,\omega')$ is an ordinary (non-graded) symplectic manifold.

In local coordinates, the auxiliary superfield (\ref{aux superfield}) is:
$$Y^i(u,du)= Y^i_{(0)}(u)+Y^i_{(1)}(u)\,du $$
where $u$ is the coordinate on $\gamma= S^1$ and index $i$ corresponds to local coordinates on $\NN$. For the pre-observable (\ref{AKSZ pre-observable}) we have
\begin{multline}
\FF_\gamma=\Map(T[1]\gamma,\NN)\\
=\left\{(Y_{(0)},Y_{(1)}) \;|\; Y_{(0)}:\gamma\ra \NN,\; Y_{(1)}\in\Gamma(\gamma, T^*\gamma\otimes Y_{(0)}^*T\NN )[-1]\right\},\\
\Omega_\gamma=-\oint_\gamma \omega'_{ij}(Y_{(0)})\; \delta Y_{(0)}^i \wedge \delta Y_{(1)}^j+ \frac{1}{2} Y_{(1)}^k\;\dd_k\omega'_{ij}(Y_{(0)})\; \delta Y_{(0)}^i \wedge \delta Y_{(0)}^j,\\
S_\gamma=\oint_\gamma \alpha'_{i}(Y_{(0)}) \; dY_{(0)}^i + \Theta'(\ii^*X,Y)
\end{multline}

One natural choice of Lagrangian in (\ref{AKSZ observable via integral over aux}) is to set
\be \LL=\Map(\gamma,\NN)\subset \Map(T[1]\gamma,\NN) \label{L for 1d obs}\ee
I.e. the 1-form component $Y_{(1)}$ of the auxiliary superfield vanishes on $\LL$. With this choice (\ref{AKSZ observable via integral over aux}) becomes
\begin{multline}\label{1D obs via PI}
O_\gamma=\int_{\Map(\gamma,\NN)} \mathcal DY_{(0)} \;\; e^{i \oint_\gamma  \alpha'_{i}(Y_{(0)}) \; dY_{(0)}^i + \Theta'(\ii^*X,Y_{(0)})} \\
=\int_{\Map(\gamma,\NN)} \mathcal DY_{(0)} \;\; e^{i \oint_\gamma  \alpha'_{i}(Y_{(0)}) \; dY_{(0)}^i + \ii^*X_{(1)}^a\;\dd_a\Theta'(\ii^*X_{(0)},Y_{(0)})}
\end{multline}
Note that with this gauge fixing, one formally has a strict version of property (\ref{obs expected iv}), namely that $O_\gamma$ is invariant under isotopy of $\gamma$ (instead of only the class of $O_\gamma$ in $Q_\Sigma$-cohomology being invariant). This follows from invariance of $\LL$ under $\Diff(\gamma)$, together with $\Diff(\gamma)$-invariance of the path integral measure $\DD Y_{(0)}$.

The quantum mechanical path integral (\ref{1D obs via PI}) can be understood in the Hamiltonian formalism, using the fiber geometric quantization of the target Hamiltonian $Q$-bundle, as follows.
\begin{prop} \label{prop 5}
Assume that the symplectic manifold $(\NN,\omega')$ can be geometrically quantized to a complex vector space of states
$\HH$ and the Hamiltonian $\Theta'\in C^\infty(\MM\times\NN)$ can be quantized (viewing coordinates on $\MM$ as external parameters) to an operator-valued function on $\MM$, $\hat\Theta'\in C^\infty(\MM)\otimes \End(\HH)$ satisfying
\be Q\hat\Theta' +i\, (\hat\Theta')^2=0 \label{hat Theta' equation}\ee
(the quantum counterpart of the equation (\ref{fiber CME})).
Then define $O_\gamma$ as the trace of a path ordered exponential\footnote{To define the path ordered exponential, we need to choose a starting point $p\in\gamma$, but the dependence on $p$ is cancelled by taking the trace.}
\be O_\gamma= \tr_\HH\; \Pexp\left( i\oint_\gamma \hat\Theta'(\ii^*X)\right) \label{1D obs Pexp}\ee
Then
$$Q_\Sigma O_\gamma=0$$
i.e. $O_\gamma$ is an observable.
\end{prop}

\begin{proof} Let $j: [0,1]\ra \Sigma$ be a path in $\Sigma$ parameterized by $t\in [0,1]$. Let us denote $\psi=\hat\Theta'(j^*X) \in \Omega^\bt([0,1])\hat\otimes C^\infty(\FF_\Sigma)\otimes \End(\HH)$, also let $\psi_{(0)}(t)$ and $\psi_{(1)}(t,dt)$ be the 0- and 1-form components of $\psi$. Consider the path ordered exponential
\begin{multline} W_j = \Pexp \left(i\int_0^1 \psi_{(1)}\right) \\
=  \lim_{N\ra\infty} \overleftarrow{\prod_{0\leq k< N}}
\left(\id_\HH+i\, \iota_{\frac{1}{N}\dd_t}\psi_{(1)}(k/N,dt)\right)
\in C^\infty(\FF_\Sigma)\otimes \End(\HH)
\end{multline}
Then
{
\begin{multline} \label{Q_Sigma W_j}
Q_\Sigma W_j = -i\int_0^1  \Pexp \left(i\int_t^1 \psi_{(1)}\right) \cdot Q_\Sigma \psi (t,dt) \cdot \Pexp \left(i\int_0^t \psi_{(1)}\right)\\
=-i\int_0^1  \Pexp \left(i\int_t^1 \psi_{(1)}\right)
\cdot \left(dt\, \frac{\dd}{\dd t}\psi_{(0)}(t)-i \left[\psi_{(0)}(t),\psi_{(1)}(t,dt)\right]\right)
\cdot\Pexp \left(i\int_0^t \psi_{(1)}\right)\\
= -i \lim_{N\ra\infty} \sum_{l=0}^{N-1}  \overleftarrow{\prod_{l< k< N}}
\left(\id_\HH+i\, \iota_{\frac{1}{N}\dd_t}\psi_{(1)}\left(\frac{k}{N},dt\right)\right)\cdot \\
\cdot \left(\psi_{(0)}\left(\frac{l+1}{N}\right) \left(\id_\HH+i\, \iota_{\frac{1}{N}\dd_t}\psi_{(1)}\left(\frac{l}{N},dt\right)\right) -
\left(\id_\HH+i\, \iota_{\frac{1}{N}\dd_t}\psi_{(1)}\left(\frac{l}{N},dt\right)\right) \psi_{(0)}\left(\frac{l}{N}\right)  \right) \cdot \\
\cdot \overleftarrow{\prod_{0\leq k< l}}
\left(\id_\HH+i\, \iota_{\frac{1}{N}\dd_t}\psi_{(1)}\left(\frac{k}{N},dt\right)\right)  \\
=-i\left(\psi_{(0)}(1)\cdot W_j- W_j\cdot \psi_{(0)}(0)\right)
\end{multline}
}
where we used that due to (\ref{hat Theta' equation}), we have
$Q_\Sigma\psi=d\psi-\frac{i}{2}[\psi,\psi]$.
For $j=i: S^1\ra\Sigma$, we have $j(0)=j(1)$ and applying (\ref{Q_Sigma W_j}), we get
$$Q_\Sigma O_\gamma = \tr_\HH\; Q_\Sigma W_j= -i\, \tr_\HH\; [\hat\Theta'(X_{(0)}(j(0))), W_j] = 0 $$
\end{proof}
Note that the construction of proposition \ref{prop 5} does not require $\omega'$ to be exact, it is enough to require that $\omega'$ satisfies Bohr-Sommerfeld integrality condition.

\subsection{Torsion-like observables}
\label{sec: torsion observables}
Now let again $\ii:\gamma\ra\Sigma$ be an oriented closed manifold of arbitrary dimension mapped into $\Sigma$ and suppose that $\NN$ is a finite-dimensional graded vector space, $\omega'$ is a graded-antisymmetric pairing on $\NN$ of degree $\dim\Sigma-1$ and $\Theta'$ is quadratic in $\NN$ directions:
$$\Theta'(x,y)=\frac{1}{2}\,\omega'(y,\theta(x)y)$$
with $\theta\in C^\infty(\MM)\otimes \End(\NN)$ of degree $1$. Equation (\ref{fiber CME}) is then equivalent to
\be Q\theta-\theta^2=0  \label{torsion obs target CME}\ee
We also assume that
\be \mr{Str}_\NN\theta(x)=0 \label{Str theta=0}\ee
where $\mr{Str}_\NN$ stands for the super-trace over $\NN$.

One way to make sense of the expression (\ref{AKSZ observable via integral over aux}) is to choose a Riemannian metric on $\gamma$.
Then we have the Hodge decomposition for differential forms on $\gamma$ with values in $\NN$
\be \underbrace{\NN\otimes \Omega^\bt(\gamma)}_{\FF_\gamma}= \underbrace{\NN\otimes \Omega_\mr{harm}^\bt(\gamma)}_{\FF_\gamma^\z}\oplus \underbrace{\NN\otimes \Omega^\bt_\mr{exact}(\gamma)\oplus \overbrace{\NN\otimes \Omega^\bt_\mr{co-exact}(\gamma)}^{\tilde\LL}}_{\tilde\FF_\gamma} \label{Hodge decomp}\ee
which can be used to first take the BV integral over the complement of harmonic forms, as in proposition \ref{prop: pushing pre-obs to zero-modes}, an then take the BV integral over the zero-modes (harmonic forms):
\be O_\gamma= \int_{\LL_\z\subset \FF_\gamma^\z}\DD Y_\z\int_{\tilde \LL\subset \tilde\FF_\gamma}\DD \tilde Y \; e^{\frac{i}{2}\int_\gamma \omega'(Y,dY+\theta(\ii^*X) Y)} \label{torsion obs integral}\ee
where $Y=Y_z+\tilde Y$ is the splitting of the auxiliary superfield into the harmonic part and the part in the complement; $\FF_\gamma^\z$ comes with the degree $-1$ symplectic structure $\Omega_\gamma^\z$ 
induced from $\omega'$ and the Poincar\'e pairing on cohomology of $\gamma$;
$\LL_\z\subset \FF_\gamma^\z$ is a Lagrangian vector subspace w.r.t. $\Omega_\gamma^\z$.

The internal integral in (\ref{torsion obs integral}) is a path integral and can be made sense of using perturbation theory\footnote{We are working modulo constants here. We are also implicitly assuming convergence of the sum over $k$ in (\ref{torsion obs on 0-modes}).}:
\begin{multline} \label{torsion obs on 0-modes}
O_\gamma^\z(X,Y_\z)
=\exp \left(\frac{i}{2} \int_\gamma\omega'(Y_\z,\theta(\ii^*X)\; (\id+ G\; \theta(\ii^*X))^{-1} Y_\z)+\right.\\
\left.+\mr{Str}_\NN\sum_{k\geq 2}\frac{(-1)^k}{2k} \int_{\mr{Conf}_k(\gamma)}\ola{\prod_{1\leq j\leq k}} \pi_{j+1,j}^*\eta \cdot \pi_j^*\theta(\ii^*X)\right)
\end{multline}
The notations here are:
\begin{itemize}
\item $\mr{Conf}_k(\gamma)$ is the Fulton-Macpherson-Axelrod-Singer compactification of the configuration space of $k$ points on $\gamma$, cf. \cite{AS}.
\item $G=d^*(\Delta+\PP_\mr{harm})^{-1}:\Omega^\bt(\gamma)\ra \Omega^{\bt-1}(\gamma)$ is the Hodge-theoretic inverse of $d$; its integral kernel extends to a smooth $(\dim\gamma-1)$-form $\eta$ on $\mr{Conf}_2(\gamma)$ (the \textit{Hodge propagator}), cf. \cite{AS}. By $\PP_\mr{harm}$ we denote the orthogonal projection to harmonic forms, using the Hodge inner product.
\item $\pi_{j+1,j}: \Conf_k \ra \Conf_2 $ is the map associated to forgetting all points except points $j$ and $j+1$; also $\pi_j: \Conf_k\ra \Conf_1=\gamma$ is forgetting all points except the point $j$. By convention, $\pi_{k+1,k}=\pi_{1,k}$.
\item In the first term in the exponential in (\ref{torsion obs on 0-modes}), $\theta(\ii^*X)$ is understood as a multiplication operator on $\Omega^\bt(\gamma)$ and an endomorphism of $\NN$, depending on ambient fields. Expression $(\id+ G\; \theta(\ii^*X))^{-1}$ is understood as
    $\sum_{k=0}^\infty (-1)^k (G\; \theta(\ii^*X))^k$.
\end{itemize}

The BV integral over zero-modes in (\ref{torsion obs integral}) then is: 
%
\begin{multline}\label{torsion obs}
O_\gamma=\int_{\LL_\z\subset \FF_\gamma^\z} e^{\frac{i}{2}\int_\gamma\omega'(Y_\z,\theta(\ii^*X)\;(\id+G\;\theta(\ii^*X))^{-1}Y_\z)}\sqrt{\mu_\z}|_{\LL_\z}\cdot \\
\cdot\exp\left(\mr{Str}_\NN\sum_{k\geq 2}\frac{(-1)^k}{2k} \int_{\mr{Conf}_k(\gamma)}\ola{\prod_{1\leq j\leq k}} \pi_{j+1,j}^*\eta \cdot \pi_j^*\theta(\ii^*X)\right)
\end{multline}
where $\mu_\z$ is some fixed translation-invariant volume element on $\FF_\z$.

\begin{rem}\label{rem: well-definedness of int over 0-modes} The requirement that the integral over zero-modes in (\ref{torsion obs}) should converge imposes the following restriction on $\LL_\z$. If $\{Y_{\LL_\z}^I\}$ are coordinates on $\LL_\z$ and if we write
$$\int_\gamma\omega'(Y_\z,\theta(\ii^*X)\;(\id+G\;\theta(\ii^*X))^{-1}Y_\z)=Y_{\LL_\z}^I H_{IJ}(\ii^*X) Y_{\LL_\z}^J$$
then we require the even-even block (i.e. the block pairing coordinates $Y_{\LL_\z}^I$ of even degree with coordinates $Y_{\LL_\z}^J$ of even degree) of the matrix $(H_{IJ})$ to be invertible.
One universal way to satisfy this condition is to choose $\LL_\z\subset \FF_\gamma^\z$ by setting all even coordinates on $\FF_\gamma^\z$ to zero. With this choice, the integral over $\LL_\z$ in (\ref{torsion obs}) is (up to normalization) the Pfaffian of $(H_{IJ})$.
\end{rem}

\begin{prop} \label{prop 6}
\begin{enumerate}[(i)]
\item \label{prop 6 item 1} Expression (\ref{torsion obs}) is an observable, i.e.
$$Q_\Sigma O_\gamma=0$$
\item \label{prop 6 item 2} Dependence of $O_\gamma$ on deformations of Riemannian metric on $\gamma$ and on deformations of $\LL_\z$ is $Q_\Sigma$-exact.
\end{enumerate}
\end{prop}

\begin{proof}
Define $S_\gamma^\z\in C^\infty(\FF_\Sigma\times \FF_\gamma^\z)$ by (\ref{torsion obs on 0-modes}) and $O^\z_\gamma=e^{iS^\z_\gamma}$:
\begin{multline}\label{S_gamma^z}
S_\gamma^\z=\frac{1}{2} \int_\gamma\omega'(Y_\z,\theta\; (\id+ G\; \theta)^{-1} Y_\z)-\\
-i\;\mr{Str}_\NN\sum_{k\geq 2}\frac{(-1)^k}{2k} \int_{\mr{Conf}_k(\gamma)}\ola{\prod_{1\leq j\leq k}} \pi_{j+1,j}^*\eta \cdot \pi_j^*\theta
\end{multline}
where $\theta$ actually means $\theta(\ii^*X)$.

First let us check that $S_\gamma^\z$
is a semi-quantum pre-observable, i.e. satisfies the equation (\ref{pre-observable QME}):
\be Q_\Sigma S_\gamma^\z+ \frac{1}{2}\{S_\gamma^\z,S_\gamma^\z\}_{\Omega_\gamma^\z}-i\Delta_\z S_\gamma^\z=0 \label{S_gamma^z QME}\ee
Indeed, we have
\begin{equation} \label{(S_gamma^z,S_gamma^z)}
\frac{1}{2}\{S_\gamma^\z,S_\gamma^\z\}_{\Omega_\gamma^\z}=\frac{1}{2}\int_\gamma\omega'(Y_\z, (\id+ \theta\; G)^{-1}\theta\;\PP_\mr{harm}\;\theta\; (\id+ G\; \theta)^{-1} Y_\z)
\end{equation}
\begin{equation}\label{Delta S_gamma^z}
-i\Delta_\z S_\gamma^\z=-\frac{i}{2}\; \mr{Str}_{\FF_\gamma^\z}\left( \PP_\mr{harm}\;\theta\; (\id+ G\; \theta)^{-1}\right)
\end{equation}
\begin{multline} \label{Q_Sigma S_gamma^z}
Q_\Sigma S_\gamma^\z=-\frac{1}{2}\int_\gamma \omega'(Y_\z,  (\id+\theta\, G)^{-1}\underbrace{(d\theta+\theta^2)}_{Q_\Sigma\theta(\ii^*X)}\,(\id+ G\,\theta)^{-1}Y_\z)-\\
+\frac{i}{2} \; \mr{Str}_\NN \sum_{k\geq 2}(-1)^{k} \int_{\mr{Conf}_k(\gamma)} \left(\ola{\prod_{2\leq j\leq k}} \pi_{j+1,j}^*\eta\cdot \pi_j^* \theta \right)\cdot \pi_{2,1}^*\eta\cdot\pi_1^* (d\theta+\theta^2)
\end{multline}
Using Stokes' theorem, (\ref{Q_Sigma S_gamma^z}) can be written as
\begin{multline} \label{Q_Sigma S_gamma^z continued}
Q_\Sigma S_\gamma^\z=\frac{1}{2}\int_\gamma \omega'(Y_\z, (\id+\theta\, G)^{-1}\theta\,([d,G]-\id)\,\theta\, (\id+G\,\theta)^{-1}  Y_\z)+\\
+\frac{i}{2}\;\mr{Str}_\NN\sum_{k\geq 2}(-1)^{k+\dim\gamma}\int_{\mr{Conf}_k(\gamma)}  \left(\ola{\prod_{2\leq j\leq k}} \pi_{j+1,j}^*\eta\cdot \pi_j^* \theta\right)
\cdot \pi_{2,1}^* d\eta\cdot \pi_1^*\theta+\\
+\frac{i}{2}\;\mr{Str}_\NN\sum_{k\geq 2}(-1)^k\int_{\mr{Conf}_k(\gamma)} \left(\ola{\prod_{2\leq j\leq k}} \pi_{j+1,j}^*\eta\cdot \pi_j^* \theta\right)\cdot \pi_{2,1}^*\eta\cdot \pi_1^*\theta^2 -\\
-\frac{i}{2}\;\mr{Str}_\NN\sum_{k\geq 2}\frac{(-1)^{k+k\dim\gamma}}{k}\int_{\dd\mr{Conf}_k(\gamma)}  \ola{\prod_{1\leq j\leq k}} \pi_{j+1,j}^*\eta\cdot \pi_j^* \theta
\end{multline}
Two last terms here cancel due to the contribution of principal boundary strata of $\mr{Conf}_k(\gamma)$, where points $j$ and $j+1$ collapse (and we use the property of the propagator that the fiber integral of $\eta$ over fibers of $\dd \mr{Conf}_2(\gamma)\ra\gamma$ is the constant function $-(-1)^{\dim\gamma}\mathbf{1}\in C^\infty(\gamma)$). Hidden boundary strata (i.e. those corresponding to the collapse of $\geq 3$ points) do not contribute due to the fact that the collapsed subgraph of the ``wheel'' graph either has a vertex of valence $2$, and then the integral vanishes by (anti-) symmetry of $\eta$ on $\dd \mr{Conf}_2(\gamma)$ w.r.t. the antipodal involution, or the collapsed subgraph has an isolated vertex, and then the integral vanishes for degree reasons, cf. \cite{BT}.

Since $[d,G]=\id-\PP_\mr{harm}$ and since $d\eta=-(-1)^{\dim\gamma}\mathbf{P}_\mr{harm}$ where $\mathbf{P}_\mr{harm}\in \Omega^{\dim\gamma}(\gamma\times\gamma)$ is the integral kernel of $\PP_\mr{harm}$, collecting (\ref{(S_gamma^z,S_gamma^z)},\ref{Delta S_gamma^z},\ref{Q_Sigma S_gamma^z continued}) we get the semi-quantum master equation (\ref{S_gamma^z QME}).

Next, knowing that (\ref{S_gamma^z QME}) holds, we infer by proposition \ref{prop: pre-obs to obs} that $O_\gamma=\int_{\LL_\z\subset \FF_\gamma^\z}e^{i S_\gamma^\z}$ given by (\ref{torsion obs}) satisfies $Q_\Sigma O_\gamma=0$. Thus we proved (\ref{prop 6 item 1}).

To prove (\ref{prop 6 item 2}) we first note that deformations of $\LL_\z$ induce $Q_\Sigma$-exact shifts in $O_\gamma$ by proposition \ref{prop: pre-obs to obs}.

Next we check the dependence of $S_\gamma^\z$ on metric on $\gamma$.
Let us now be more careful with the space of zero-modes and define it as $\FF_\gamma^\z=\NN\otimes H^\bt(\gamma)$ -- the $\NN$-valued de Rham cohomology of $\gamma$; it comes with a metric-dependent embedding of cohomology into differential forms $\iii:H^\bt(\gamma)\hra  \Omega^\bt(\gamma)$ which realizes cohomology classes by harmonic forms. In view of this redefinition we should now write $\iii(Y_\z)$ instead of $Y_\z$ in (\ref{S_gamma^z}).

Now assume that we have a path in the space of Riemannian metrics on $\gamma$, $g_t\in \mr{Met}(\gamma)$ with $t\in [0,1]$ the parameter. Then we have the $t$-dependent Hodge-theoretic inverse $G_t$ for de Rham operator, the $t$-dependent propagator $\eta_t$ and the $t$-dependent embedding of cohomology into differential forms $\iii_t$ with exact time-derivatives:
$$\frac{d}{dt} G_t= [d, H_t], \quad \frac{\dd}{\dd t} \eta_t = d\zeta_t,\quad \frac{\dd}{\dd t}\iii_t=d j_t$$
for some linear operator $H_t:\Omega^\bt(\gamma)\ra \Omega^{\bt-2}(\gamma)$ with smooth integral kernel $\zeta_t\in \Omega^{\dim\gamma-2}(\gamma\times\gamma)$ and some $j_t: H^\bt(\gamma)\ra \Omega^{\bt-1}(\gamma)$; $H_t$, $\zeta_t$ and $j_t$ depend on $g_t$ and $\frac{\dd}{\dd t}g_t$. Then one can introduce the ``extended'' versions of $G$, $\eta$ and $\iii$:
\begin{eqnarray*}
\tilde G = G_t+ dt\; H_t &\in& \Omega^\bt([0,1])\hat\otimes \End(\Omega^\bt(\gamma))\\
\tilde \eta = \eta_t+ dt\; \zeta_t &\in& \Omega^{\dim\Sigma-1}([0,1]\times \mr{Conf}_2(\gamma))\\
\tilde {\iii} = \iii_t+dt\; j_t &\in& \Omega^\bt([0,1])\hat\otimes \mr{Hom}(H^\bt(\gamma),\Omega^\bt(\gamma))
\end{eqnarray*}
They satisfy respectively
\begin{eqnarray*}
\left(dt\frac{\dd}{\dd t}+[d_\gamma,\bt]\right)\tilde G &=& \id-\PP_\mr{harm}\\
\left(dt\frac{\dd}{\dd t}+d_{\mr{Conf}_2(\gamma)}\right)\tilde\eta &=& -(-1)^{\dim\gamma}\mathbf{P}_\mr{harm}\\
\left(dt\frac{\dd}{\dd t}+d_\gamma\right)\tilde {\iii} &=& 0
\end{eqnarray*}
Next, we define the extended version of $S_\gamma^z$,  $\tilde S_\gamma^\z\in \Omega^\bt([0,1])\hat\otimes C^\infty(\FF_\Sigma\times \FF_\gamma^\z)$ by substituting $\tilde G$ instead of $G$, $\tilde\eta$ instead of $\eta$ and $\tilde {\iii}(Y_\z)$ instead of $Y_\z$ into (\ref{S_gamma^z}). Repeating our proof of (\ref{S_gamma^z QME}) in the extended case we find that
$$\left(dt\frac{\dd}{\dd t}+Q_\Sigma\right)\tilde S_\gamma^\z+\frac{1}{2}\{\tilde S_\gamma^\z,\tilde S_\gamma^\z\}_{\Omega_\gamma^\z}-i\Delta_\z \tilde S_\gamma^\z=0$$
which, when we decompose $\tilde S_\gamma^z$ into the 0-form and 1-form parts w.r.t. $t$, $\tilde S_\gamma^\z=S_\gamma^\z(t)+dt\, R_\gamma^\z(t)$,  is equivalent to semi-quantum master equation (\ref{S_gamma^z QME}) for $S_\gamma^\z(t)$ plus the equation
\be \frac{\dd}{\dd t} S_\gamma^\z(t)= Q_\Sigma R_\gamma^\z(t)+\{S_\gamma^\z(t),R_\gamma^\z(t)\}_{\Omega_\gamma^\z}-i \Delta_\z R_\gamma^\z(t) \ee
Therefore (cf. (\ref{equivalence of sq pre-obs, infinitesimal})) semi-quantum pre-observables $S_\gamma^\z(t)$ are equivalent for all values of the parameter $t$ and so by item (\ref{prop 1 item 3}) of proposition \ref{prop: pre-obs to obs}, dependence of the induced observable $O_\gamma(t)=\int_{\LL_\z\subset \FF_\gamma^\z} e^{i S_\gamma^\z(t)}$ on $t$ is $Q_\Sigma$-exact. This finishes the proof of (\ref{prop 6 item 2}).
\end{proof}


\section{Examples of observables}\label{sec: examples of observables}
Here we will illustrate the construction of AKSZ pre-observables/observables of sections \ref{sec: AKSZ pre-observable}, \ref{sec: from pre-obs to obs} by several explicit examples.

\subsection{Observables with $\NN$ a point.}
First we give some examples of the situation described in section \ref{sec: obs with N=pt}.
\begin{itemize}
\item For abelian Chern-Simons theory (\ref{abelian CS}) we may take $\theta=c \psi \in C^\infty(\RR[1])$ with $c\in\RR$ any constant. This yields by (\ref{N=pt obs}) a 1-dimensional observable for $\ii:\gamma=S^1\hra \Sigma$
    $$O_\gamma=e^{i c \oint_\gamma \ii^*A_{(1)}}$$
    --- the abelian Wilson loop; parameter $c$ can be viewed as the weight of a one-dimensional representation of the gauge group $\RR$. Note that this is also an observable for the abelian $BF$ theory (\ref{abelian BF}).
\item For $D$-dimensional abelian $BF$ theory, take $\theta=c\xi$ with $c\in \RR$ a parameter. The associated observable for a closed oriented submanifold $\ii:\gamma\hra\Sigma$ of codimension 2 is
    \be O_\gamma=e^{i c \int_\gamma \ii^*B_{(D-2)}} \label{Wilson B surface in abBF}\ee
\item For 2-dimensional $BF$ theory assume that $\g$ is quadratic so that we can identify $\g^*$ with $\g$ and let $\rho:\g\ra\End(R)$ be a representation. Set $\theta=-i\log\tr_{R} f(\rho(\xi))$ for $f\in C^\infty(\RR)$ a function.
    The corresponding 0-dimensional observable for $\gamma\in \Sigma$ a point is
    \be O_\gamma=\tr_R f(\rho(B_{(0)}|_\gamma)) \label{obs ex: B-obs for 2d BF}\ee
\item For 2-dimensional abelian $BF$ theory take $\theta=f(\xi)\psi$ for $f\in C^\infty(\RR)$. The corresponding 1-dimensional observable for $\ii:\gamma=S^1\hra \Sigma$ is
    \be O_\gamma=e^{i\oint_\gamma \ii^* (A \cdot f(B))} \label{obs ex: 1d obs for 2d ab BF}\ee
    Note that for $f$ non-constant, $O_\gamma$ depends non-trivially on components of the superfields of ghost number $\neq 0$, namely on $A_{(0)}$ and  $B_{(1)}$:
    $$ O_\gamma=e^{i \oint_\gamma \ii^*(A_{(1)}\cdot f(B_{(0)})+ A_{(0)}\cdot B_{(1)}\cdot f'(B_{(0)})) } $$
\item For Poisson sigma model and $\theta=f\in C^\infty(M)$ a Casimir of the Poisson structure (i.e. $[\pi,f]=0$, with $[,]$ the Schouten-Nijenhuis bracket on polyvector fields on $M$), we get a 0-dimensional observable for $\gamma$ a point in $\Sigma$:
    \be O_\gamma= e^{i\, f(X_{(0)}|_\gamma)} \label{obs ex: PSM 0d obs}\ee
\item For Poisson sigma model and $v\in \Vect(M)$ a Poisson vector field (i.e. $[\pi,v]=0$), set $\theta=\langle p, v(x)\rangle$. The associated 1-dimensional observable is:
    \be O_\gamma=e^{i\oint_\gamma \ii^*\langle \eta, v(X)\rangle} \label{obs ex: PSM 1d obs}\ee
    Note that if $\tilde v=v+[\pi,f]$ -- another Poisson vector field giving the same class in Poisson cohomology of $M$ as $v$, then the observable $\tilde O_\gamma$ constructed from $\tilde v$ by (\ref{obs ex: PSM 1d obs}) is equivalent to $O_\gamma$:
    $\langle p,\tilde v(x) \rangle - \langle p,v(x)\rangle= Q f(x)$ on the target implies $\oint_\gamma \ii^*\langle\eta ,\tilde v (X)\rangle - \oint_\gamma \ii^*\langle\eta ,v (X)\rangle=- Q_\Sigma \oint_\gamma \ii^*f(X)$ on the mapping space $\FF_\Sigma$. This implies in turn
    $$\tilde O_\gamma= O_\gamma+ Q_\Sigma \left( -i \int_0^1 dt\; e^{i\oint_\gamma \ii^*\langle \eta, (1-t) v(X)+ t \tilde v(X) \rangle} \oint_\gamma \ii^*f(X)\right)$$
    Thus (\ref{obs ex: PSM 1d obs}) defines a map from Poisson cohomology of $M$ in degree one to $Q_\Sigma$-cohomology in degree zero.
\end{itemize}

Note that the observable (\ref{obs ex: B-obs for 2d BF}) is a special case of (\ref{obs ex: PSM 0d obs}) when we take $M$ to be $\g^*$ with Kirillov-Kostant Poisson structure. Also, (\ref{obs ex: 1d obs for 2d ab BF}) is a special case of (\ref{obs ex: PSM 1d obs}) for $M=\RR$ with $\pi=0$.

\subsection{Wilson loop observable in Chern-Simons theory}
Here we will combine the discussion of section \ref{sec: 1D obs} with the example (\ref{Ham Q-bun ex Ham action}) of section \ref{sec: examples of Ham Q-bun}.

Let $\g$ be a quadratic Lie algebra and let $(M,\omega_M)$ be a symplectic manifold with a Hamiltonian action of $\g$ with equivariant moment map $\mu: M\ra \g^*$.
Further, we assume that $\omega_M$ satisfies Bohr-Sommerfeld integrality condition and is the curvature of a $U(1)$-connection $\nabla$ on a line bundle $L$ over $M$.

Then we have a degree $0$ Hamiltonian $Q$-bundle structure on $\g[1]\times M\ra \g[1]$, with base structure given by (\ref{CS target}) and fiber structure (\ref{fiber data from Ham action of g}), with the correction that now we relaxed the exactness condition on $\omega_M$. Applying the construction of section \ref{sec: AKSZ pre-observable} to this target data we get the following $1$-dimensional pre-observable for Chern-Simons theory:
\begin{multline} \label{Wilson loop pre-obs}
\FF_\gamma= \Map(T[1]\gamma,M)\\
=\{(Y_{(0)},Y_{(1)})\;|\; Y_{(0)}: \gamma\ra M,\; Y_{(1)}\in \Gamma(\gamma,T^*\gamma\otimes Y_{(0)}^*TM)[-1]\},\\
\Omega_\gamma=-\oint_\gamma \omega_M (Y)=
-\oint_\gamma \omega_{ij}(Y_{(0)}) \delta Y_{(0)}^i\wedge \delta Y_{(1)}^j+\frac{1}{2} Y_{(1)}^k \dd_k \omega_{ij}(Y_{(0)}) \delta Y_{(0)}^i\wedge \delta Y_{(0)}^j,\\
e^{i S_\gamma}= \mr{Hol}_{Y_{(0)}^*\nabla}(\gamma)\; e^{i\oint_\gamma \langle \ii^*A,\mu(Y)\rangle}=
 \mr{Hol}_{Y_{(0)}^*\nabla}(\gamma)\; e^{i\oint_\gamma \langle \ii^* A_{(1)},\mu(Y_{(0)})\rangle + \langle \ii^*A_{(0)}, Y_{(1)}^i\dd_i\mu(Y_{(0)})\rangle}
\end{multline}
where $\ii:\gamma=S^1\hra \Sigma$ is a circle embedded into 3-manifold $\Sigma$, $A=\sum_{k=0}^3 A_{(k)}$ is the superfield of the ambient Chern-Simons theory, $Y=Y_{(0)}+Y_{(1)}$ is the auxiliary superfield (\ref{CS superfield}), with components $Y_{(0)}$, $Y_{(1)}$ having internal degrees $0$ and $-1$ respectively; indices $i,j,k$ refer to local coordinates on $M$.

To construct an observable out of the pre-observable (\ref{Wilson loop pre-obs}), we impose the gauge fixing (\ref{L for 1d obs}) which consists in setting $Y_{(1)}=0$. The path integral (\ref{1D obs via PI}) reads
\be O_\gamma= \int_{\Map(\gamma,M)}\DD Y_{(0)}\; \mr{Hol}_{Y_{(0)}^*\nabla}(\gamma)\; e^{i\oint_\gamma \langle \ii^*A_{(1)},\mu(Y_{(0)})\rangle}  \label{Wilson loop PI}\ee

Assume that there exists a Lagrangian polarization $P$ of $(M,\omega_M)$ such that the geometric quantization of $(M,\omega,L,\nabla,P)$ yields a space of states $\HH$ and components $\mu_a\in C^\infty(M)$ of the moment map are quantized into operators $\hat\mu_a\in \End(\HH)$ such that the map $\hat\mu: \g\ra \End(\HH)$ sending generator $T_a\in\g$ into $\hat\mu_a$ is a Lie algebra homomorphism (i.e. $\hat\mu$ is a representation of $\g$ on $\HH$). Then
$$\hat\Theta'=-i\langle \psi, \hat\mu\rangle \in C^\infty(\g[1])\otimes \End(\HH)$$
(where we view $\hat\mu$ as an element of $\g^*\otimes\End(\HH)$) is the quantization of $\Theta'=\langle \psi,\mu(x) \rangle$ satisfying (\ref{hat Theta' equation}). Then the construction of proposition \ref{prop 5} produces out of the pre-observable (\ref{Wilson loop pre-obs}) the usual Wilson loop of the connection component $A_{(1)}$ of the ambient superfield in representation $\hat\mu$ along $\gamma$:
\be O_\gamma=\tr_{\HH} \Pexp \oint_\gamma \hat\mu(\ii^*A_{(1)}) \label{Wilson loop Pexp}\ee
which is indeed a gauge invariant observable (i.e. $Q_\Sigma O_\gamma=0$) for Chern-Simons theory on $\Sigma$.

Expression (\ref{Wilson loop Pexp}) can be viewed as a regularization of
the path integral (\ref{Wilson loop PI}).

\textit{Coadjoint orbit case.}
In case when $(M,\omega_M)$ is an integral coadjoint orbit of a compact Lie group $G$ such that $\g=\mr{Lie}(G)$, with $\omega_M$ the Kirillov symplectic structure and $\mu:M\hra \g^*$ the embedding of $M$ into $\g^*$ (which corresponds to the coadjoint Hamiltonian action of $G$ on $M$), formula (\ref{Wilson loop PI}) becomes the Alekseev-Faddeev-Shatashvili path integral formula for the Wilson loop \cite{AFS}. The respective $\hat\mu$ is the irreducible representation of $G$ associated by Kirillov's orbit method to the coadjoint orbit $M$, and (\ref{Wilson loop Pexp}) is just the Wilson loop in this irreducible representation.

\subsection{``Wilson loops'' in Poisson sigma model} \label{sec: Wilson loop in PSM}
Let $(M,\pi)$ be a Poisson manifold, $(N,\omega_N=\delta \alpha_N)$ an exact symplectic manifold and $F: N\ra \Vect(M)$ a map\footnote{Alternatively, one can view $F$ as a vertical vector field on the trivial fiber bundle $N\times M\ra N$.} with the property
\be \frac{1}{2}\{F,F\}_{\omega_N}+ [\pi,F]=0 \label{PSM Wilson loop eq on F}\ee
where $[,]$ is the Schouten-Nijenhuis bracket on polyvector fields on $M$; expression $\frac{1}{2}\{F,F\}$ is a map from $N$ to bivectors on $M$, given in local coordinates as $\frac{1}{2}(\omega^{-1}(y))^{ab}\,\frac{\dd}{\dd y^a} F^i(x,y)\, \frac{\dd}{\dd y^b} F^j(x,y)\; \frac{\dd}{\dd x^i} \wedge \frac{\dd}{\dd x^j}$ where $\{x^i\}$ are local coordinates on $M$ and $\{y^a\}$ are local coordinates on $N$. Then we have a degree $0$ Hamiltonian $Q$-bundle structure on $T^*[1]M\times N \ra T^*[1]M$ with base structure (\ref{PSM target data}) and fiber structure
\begin{multline}
\NN=N,\quad \A=\langle p, \{F,\bt\}_{\omega_N}  \rangle,  \\ \omega'=\omega_N,\quad\alpha'=\alpha_N,\quad  \Theta'= \langle p, F \rangle \in C^\infty(T^*[1]M\times N)
\end{multline}

The corresponding 1-dimensional pre-observable for Poisson sigma model associated to $(M,\pi)$ is given by $\FF_\gamma$, $\Omega_\gamma$ as in (\ref{Wilson loop pre-obs}), changing $M$ to $N$, and the auxiliary action is
$$S_\gamma= \oint_\gamma  \alpha_a(Y)\, dY^a+\langle \ii^*\eta, F(\ii^*X,Y) \rangle$$

The corresponding observable is formally given (in the gauge (\ref{L for 1d obs})) by the path integral (\ref{1D obs via PI}):
\begin{equation}\label{PSM Wilson loop PI}
O_\gamma=\int_{\Map(\gamma,N)} \DD Y_{(0)}\; e^{i\oint_\gamma \alpha_a(Y_{(0)})\, dY_{(0)}^a+\langle \ii^*\eta, F(\ii^*X,Y_{(0)})\rangle}
\end{equation}

Assume that $(N,\omega_N)$ can be geometrically quantized to a space of states $\HH$ and $F$ is quantized to an operator-valued vector field $\hat{F}\in \End(\HH)\hat\otimes \Vect(M)$ satisfying
\be [\pi,\hat F]+ i \hat F\wedge \hat F=0 \label{PSM Wilson loop eq on F hat}\ee
Then (\ref{1D obs Pexp}) gives
\be O_\gamma= \tr_\HH \Pexp \left(i\oint_\gamma \langle \ii^*\eta ,\hat{F}(\ii^*X)\rangle \right) \label{PSM Wilson loop Pexp}\ee
which can be regarded as the evaluation of the path integral (\ref{PSM Wilson loop PI}) via Hamiltonian formalism. By proposition \ref{prop 5}, $O_\gamma$ is indeed an observable for the Poisson sigma model.

\begin{rem} \label{rem: PSM Wilson loop}
Note that in case $N=$ point we get back the observable (\ref{obs ex: PSM 1d obs}). Also, in case $M=\g^*$ with Kirillov-Kostant Poisson structure and requiring that $F$ takes values in constant vector fields on $F$, equation (\ref{PSM Wilson loop eq on F}) becomes the equation on equivariant moment map and the corresponding observable (\ref{PSM Wilson loop Pexp}) is the usual Wilson loop (\ref{Wilson loop Pexp}) for 2-dimensional $BF$ theory.
\end{rem}

\subsection{Torsion observables in Chern-Simons theory}
Fix $n'\in \{-1,0,1,2\}$ and $m\in \ZZ$ (only the parity of $m$ will matter for our discussion).
Let again $\g$ be a quadratic Lie algebra and let $\rho: \g\ra \End(R)$ be a representation on a vector space $R$ with values in traceless matrices (cf. (\ref{Str theta=0})). Then we have a degree $n'$ Hamiltonian $Q$-bundle structure on
$\g[1]\oplus R[m]\oplus R^*[n'-m]\ra \g[1]$ with base structure (\ref{CS target}) and fiber structure
\begin{multline}
\NN=R[m]\oplus R^*[n'-m],\quad
\A=\left\langle \rho(\psi)q,\frac{\dd}{\dd q}\right\rangle+ \left\langle \rho^*(\psi)p,\frac{\dd}{\dd p} \right\rangle,\\
\omega'=\langle \delta p,\delta q \rangle,\quad \alpha'=\langle p,\delta q\rangle,\quad  \Theta'=\langle p,\rho(\psi)q \rangle
\end{multline}
where $\rho^*: \g\ra \End(R^*)$ is the representation dual to $\rho$; $q$ is the $R$-valued coordinate of degree $m$ on $R[m]$ and $p$ is the $R^*$-valued coordinate of degree $n'-m$ on $R^*[n'-m]$.

The corresponding pre-observable for Chern-Simons theory on a closed oriented 3-manifold $\Sigma$, for $\ii:\gamma\hra \Sigma$ an embedded closed oriented $(n'+1)$-manifold, reads:
\begin{multline}\label{CS torsion preobs}
\FF_\gamma = R[m]\otimes \Omega^\bt(\gamma) \oplus R^*[n'-m]\otimes \Omega^\bt(\gamma),\quad
\Omega_\gamma = \int_\gamma \langle \delta \mb{p}, \delta\mb{q} \rangle,\\
S_\gamma = \int_\gamma \langle \mb{p} , d\mb{q}\rangle +\langle \mb{p}, \rho(\ii^*A)\mb{q} \rangle
\end{multline}
where $\mb{q}=\sum_{k=0}^{n'+1}\mb{q}_{(k)}$ and $\mb{p}=\sum_{k=0}^{n'+1} \mb{p}_{(k)}$ are the auxiliary superfields corresponding to $q$ and $p$; component $\mb{q}_{(k)}$ is a  $R$-valued $k$-form on $\gamma$ with internal degree $m-k$ and $\mb{p}_{(k)}$ is a $R^*$-valued $k$-form on $\gamma$ of internal degree $n'-m-k$.

Pre-observable (\ref{CS torsion preobs}) can be pushed forward to zero-modes (in the sense of proposition \ref{prop: pushing pre-obs to zero-modes}) as in section \ref{sec: torsion observables}:
\begin{multline} \label{CS torsion zero-modes}
\FF_\gamma^\z= R[m]\otimes H^\bt(\gamma) \oplus R^*[n'-m]\otimes H^\bt(\gamma), \quad
\Omega_\gamma^\z = \int_\gamma \langle \delta \mb{p}_\z, \delta \mb{q}_\z \rangle,\\
S_\gamma^\z = \int_\gamma \langle \mb{p}_\z, \rho(\ii^*A) (\id+G\; \rho(\ii^*A))^{-1} \mb{q}_\z \rangle - i \log \mr{tor}(\gamma,\ii^*A,\rho)
\end{multline}
where we introduced the notation
\be \mr{tor}(\gamma,\ii^*A,\rho)=\exp \tr_R \sum_{k\geq 2}\frac{(-1)^k}{k} \int_{\mr{Conf}_k(\gamma)} \ola{\prod_{1\leq j\leq k}} \pi_{j+1,j}^*\eta\cdot \pi_j^* \rho(\ii^*A)  \label{tor}\ee
Notations $G$, $\eta$, $\pi$ are the same as in section \ref{sec: torsion observables}. We use harmonic representatives for zero-modes $\mb{q}_\z$, $\mb{p}_\z$ in $\FF_\gamma$.

Taking the BV integral over zero-modes in (\ref{CS torsion zero-modes}) we obtain the observable for Chern-Simons theory on $\Sigma$:
\be O_\gamma=\int_{\LL_\z\subset \FF_\gamma^\z} \sqrt{\mu_\z}|_{\LL_\z}\; e^{i\int_\gamma \langle \mb{p}_\z, \rho(\ii^*A) (\id+G\; \rho(\ii^*A))^{-1} \mb{q}_\z \rangle}\cdot \mr{tor}(\gamma,\ii^*A,\rho) \label{CS torsion obs}\ee
with $\LL_\z\subset \FF_\gamma^\z$ a Lagrangian subspace and $\mu_\z$ a translation invariant volume element on $\FF_\gamma^\z$. Proposition \ref{prop 6} implies that $O_\gamma$ is $Q_\Sigma$-closed and that deformations of gauge-fixing induce $Q_\Sigma$-exact shifts in $O_\gamma$.

In case $n'=0$, $m=1$, with $\gamma= S^1$, (\ref{CS torsion obs}) can be evaluated explicitly\footnote{One trick to simplify the calculation is to use gauge-invariance of $O_\gamma$ to pick the constant representative of the gauge equivalence class of connection $\ii^* A$.}:
\begin{multline} \label{CS torsion obs 1D}
O_\gamma= \int_{R[1]\oplus R^*[-1]} d\mb{q}_{(0)}^\z\; d\mb{p}_{(0)}^\z\; e^{i\oint_\gamma\langle \mb{p}_{(0)}^\z, \rho(\overline A) \mb{q}_{(0)}^\z  \rangle} \cdot {\det}_R\left(\frac{\sinh\rho(\overline A/2)}{\rho(\overline A/2)}\right)\\
=\mr{det}_R \left(\rho(W)^{1/2}-\rho(W)^{-1/2}\right)=\mr{det}_R (\rho(W)-\id_R)
\end{multline}
where $W=\Pexp \oint_\gamma \ii^*A_{(1)}$ is the holonomy of the ambient connection around $\gamma$ (since this requires the choice of a starting point on $\gamma$, one can think of $W$ as a group element defined modulo conjugation). We also denote $\overline A= \log W$.

In (\ref{CS torsion obs 1D}) we used the gauge-fixing Lagrangian (\ref{L for 1d obs}). In particular, the gauge-fixing for the integral over zero-modes is
$$\underbrace{(R[1]\oplus R^*[-1])\otimes H^0(\gamma)}_{\LL_\z}\subset \underbrace{(R[1]\oplus R^*[-1])\otimes H^\bt(\gamma)}_{\FF^\z_\gamma} $$

\begin{rem}\label{rem: CS 1D obs condition on rho} Observable (\ref{CS torsion obs 1D}) is only interesting (not identically zero) for representations $\rho$ such that $\det \rho\not\equiv 0 $, i.e. $\rho(x)$ does not have zero eigenvalue for $x$ in an open dense subset of $\g$. In particular, (\ref{CS torsion obs 1D}) is identically zero for $\rho=\ad$ the adjoint representation.
\end{rem}

\begin{rem} Observable (\ref{CS torsion obs 1D}) has the property that it only depends on $\ii^*A_{(1)}$ and not on the other components of the ambient field. Also, it is the only representative of its class in $Q_\Sigma$-cohomology which depends only on $\ii^*A$, for degree reasons (since $C^\infty(\FF_\Sigma|_\gamma)$ is non-negatively graded; we denote $\FF_\Sigma|_\gamma=\Map(T[1]\gamma,\MM)$ the space of pull-backs of ambient fields to $\gamma$).

In case $n'> 0$, 
$O_\gamma$ defined by (\ref{CS torsion obs}) generally depends on $\ii^*A_{(k)}$ for $0\leq k\leq n'+1$. Also, there is no canonical representative of the class of $O_\gamma$ in $Q_\Sigma$-cohomology and no canonical choice for the gauge-fixing Lagrangian $\LL\subset \FF_\gamma$.
\end{rem}

For $n'=-1$, convergence requirement for the integral over zero-modes (\ref{CS torsion obs}) forces (cf. remark \ref{rem: well-definedness of int over 0-modes}) $\LL_\z=(R[m]\oplus R^*[-1-m])^\mr{odd}$ which yields $O_\gamma=0$.

For $n'=2$, for $\gamma$ a connected closed oriented 3-manifold with non-zero first Betti number, we may take $m=1$ and choose the gauge-fixing for zero-modes as
\be \LL_\z=(R[1]\oplus R^*[1])\otimes (H^0(\gamma)\oplus \underbrace{\lambda^\perp}_{\subset H^1(\gamma)} \oplus \underbrace{\lambda}_{\subset H^2(\gamma)}) \label{CS torsion obs 3D L_z}\ee
where $\lambda$ is any line in $H^1(\gamma)$ and  $\lambda^\perp$ is its orthogonal complement in $H^2(\gamma)$ w.r.t. the Poincar\'e duality.

If the condition $\det\rho\not\equiv 0$ as in remark \ref{rem: CS 1D obs condition on rho} holds, $O_\gamma$ defined with zero-mode gauge-fixing (\ref{CS torsion obs 3D L_z}) is well-defined on an open subset of $\FF_\Sigma|_\gamma$ and is not identically zero.

In case $\dim H^1(\gamma)=1$, Lagrangian (\ref{CS torsion obs 3D L_z}) is the same as in remark \ref{rem: well-definedness of int over 0-modes}, i.e. $\LL_\z=(\FF_\gamma^\z)^\mr{odd}$. In case $\dim H^1(\gamma)> 1$, the Lagrangian $(\FF_\gamma^\z)^\mr{odd}$ yields $O_\gamma$ which is identically zero.

In case of $\gamma$ a rational homology sphere, $H^1(\gamma)=0$, construction (\ref{CS torsion obs 3D L_z}) is not applicable. However, one can take
$$\LL_\z=(R[1]\oplus R^*[1])\otimes H^0(\gamma)$$
which produces an observable $O_\gamma$ of degree $-2 \dim R$.

\begin{rem}
In case $n'=2$ and $\gamma=\Sigma$, $\ii=\id$,
property (\ref{(Theta',Theta')_omega=0}) does not generally hold, so (\ref{CS torsion preobs}) is a pre-observable in the sense of definition \ref{def: class pre-observable}, but $S_\Sigma+S_\gamma$ does not satisfy the classical master equation (\ref{S_Sigma+S_gamma CME}) on the space $\FF_\Sigma\times \FF_\gamma$. On the level of observable $O_\gamma$ this means that $Q_\Sigma O_\gamma=0$, but $S_\Sigma-i\log O_\gamma$ does not satisfy the classical master equation on $\FF_\Sigma$.
\end{rem}

%

In case $n'=1$, for $\gamma$ a connected closed oriented surface of genus $g>0$, we may take $m=0$ and set
$$\LL_\z=R\otimes (H^0(\gamma)\oplus \lambda^\perp) \oplus R^*[1]\otimes (H^0(\gamma)\oplus \lambda)$$
where $\lambda$ is again any line $H^1(\gamma)$ and $\lambda^\perp$ is its orthogonal in $H^1(\gamma)$ w.r.t. the Poincar\'e duality. Assuming $\det\rho\not\equiv 0$, with this choice of $\LL_\z$ we produce an observable $O_\gamma$ of degree $(2g-2)\dim R$ which is well-defined and non-zero on an open subset of $\FF_\Sigma|_\gamma$.

Expression (\ref{CS torsion obs}) also is an observable for $BF$ theory in dimension $D$. In this case $n'$ should satisfy $-1\leq n' \leq D-1$. For $n'=2k$ and $\gamma$ a closed oriented submanifold of $\Sigma $ of dimension $2k+1$, we may take $m=2k+1$ and choose the gauge-fixing $\LL_\z$ as
$$ \LL_\z=(R[2k+1]\oplus R^*[-1])\otimes \LL'$$
where $\LL'\subset H^\bt(\gamma)$ a Lagrangian subspace with the property
\be \dim \LL'_{j}= \dim \LL'_{2k-j} \label{CS torsion obs L' dim condition}\ee where $\LL'_j=\LL'\cap H^j(\gamma)$,  which is equivalent to
\be \dim \LL'_j=B_j-B_{j-1}+\cdots+(-1)^j B_0 \label{CS torsion obs L' dim}\ee
where $B_j=\dim H^j(\gamma)$. For $\LL'$ to exist, we require that for $\gamma$ the combinations of Betti numbers on the r.h.s. of (\ref{CS torsion obs L' dim}) are non-negative.
Assuming $\det\rho\not\equiv 0$, property (\ref{CS torsion obs L' dim condition}) ensures that the integral over zero-modes in (\ref{CS torsion obs}) exists and is non-zero
in an open subset of $\FF_\Sigma|_\gamma$.

\subsection{A codimension 2 (pre-)observable in $BF$ theory}
The following example is due to Cattaneo-Rossi \cite{Cattaneo-Rossi}.
Let $\MM=\g[1]\oplus \g^*[D-2]$ with the structure of degree $D-1$ Hamiltonian $Q$-manifold as in (\ref{BF target}). One constructs a trivial degree $D-3$ Hamiltonian $Q$-bundle over $\MM$ with fiber data
\begin{multline} \label{CR target fiber for BF}
\NN=\g\oplus \g^*[D-3],\quad \A=\langle[\psi,q],\frac{\dd}{\dd q}\rangle+ \langle \ad^*_\psi p ,\frac{\dd}{\dd p} \rangle+ (-1)^D\langle\xi, \frac{\dd}{\dd p}\rangle  ,\\
\omega'=\langle \delta p, \delta q \rangle, \quad \alpha'=\langle p, \delta q\rangle, \quad \Theta'=\langle p, [\psi,q] \rangle+\langle \xi, q\rangle
\end{multline}
Here $q$ is the  $\g$-valued coordinate of degree $0$ on $\g$ and $p$ is the  $\g^*$-valued coordinate of degree $D-3$ on $\g^*[D-3]$; $\psi$ and $\xi$ are coordinates on $\MM$ as in section \ref{sec:AKSZ examples}.

Given a closed oriented $D$-manifold $\Sigma$ and a closed oriented submanifold $\ii:\gamma\hra \Sigma$ of codimension $2$, by the construction of section \ref{sec: AKSZ pre-observable} we get a pre-observable for the $BF$ theory on $\Sigma$ associated to the Hamiltonian $Q$-bundle (\ref{BF target},\ref{CR target fiber for BF}):
\begin{eqnarray}
\label{CR preobs 1}
\FF_\gamma &=& \g\otimes \Omega^\bt(\gamma)\oplus \g^*[D-3]\otimes \Omega^\bt(\gamma), \\
\Omega_\gamma&=&(-1)^D\int_\gamma \langle\delta \mb{p},\delta\mb{q} \rangle,  \label{CR preobs 2} \\
S_\gamma&=&\int_\gamma \langle \mb{p} ,d\mb{q}\rangle + \langle \mb{p}, [\ii^*A,\mb{q}] \rangle+ \langle \ii^*B, \mb{q}\rangle \label{CR preobs 3}
\end{eqnarray}
where $\mb{q}=\sum_{k=0}^{D-2} \mb{q}_{(k)}$ and $\mb{p}=\sum_{k=0}^{D-2} \mb{p}_{(k)}$ are the auxiliary superfields corresponding to $q$ and $p$; $\mb{q}_{(k)}$ and $\mb{p}_{(k)}$ are $k$-forms on $\gamma$ with values in $\g$ and $\g^*$ respectively, having internal degrees $-k$ and $D-3-k$ respectively.

Push-forward of the pre-observable (\ref{CR preobs 1}--\ref{CR preobs 3}) to zero-modes (in the sense of proposition \ref{prop: pushing pre-obs to zero-modes}) is evaluated as in section \ref{sec: torsion observables}:
\begin{multline} \label{CR preobs on 0-modes}
\FF_\gamma^\z=\g\otimes H^\bt(\gamma) \oplus \g^*[D-3]\otimes H^\bt(\gamma),\qquad
\Omega_\gamma^\z= \int_\gamma \langle \delta \mb{p}_\z, \delta \mb{q}_\z \rangle,\\
S_\gamma^\z=\int_\gamma \langle \ii^*B +(-1)^D \ad^*_{\ii^*A}\mb{p}_\z,(\id+ G\; \ad_{\ii^*A})^{-1} \mb{q}_\z \rangle
-i 
\log\mr{tor}(\gamma, \ii^*A,\ad)
\end{multline}
where $\mr{tor}$ is as in (\ref{tor}), for $\rho=\ad$ the adjoint representation of $\g$.

In the case of abelian $BF$ theory, $\g=\RR$, (\ref{CR preobs on 0-modes}) yields
\be e^{i S_\gamma^\z}=e^{i \int_\gamma \langle \ii^*B,\mb{q}_\z \rangle} \label{CR obs abelian}\ee
which is an observable if we understand $\mb{q}_\z\in H^\bt(\gamma)$ as an external parameter (the crucial point here is that by construction (\ref{CR obs abelian}) satisfies the semi-quantum master equation (\ref{pre-observable QME exp form}), but the last term in (\ref{pre-observable QME exp form}) vanishes since $S_\gamma^\z$ depends only on $\mb{q}_\z$).
Taking $\mb{q}_\z=c\cdot {1}\in H^0(\gamma)$, we obtain the observable (\ref{Wilson B surface in abBF}).

In \cite{Cattaneo-Rossi} it is shown that in the case of ``long knots'' $\RR^{D-2}\sim \gamma \subset \Sigma=\RR^D$, embeddings of $\RR^{D-2}$ into $\RR^D$ with prescribed linear asymptotics, the push-forward of pre-observable (\ref{CR preobs 1}--\ref{CR preobs 3}) to zero-modes yields, upon fixing the values of zero-modes to $\mb{q}_\z=c\otimes 1\in \g^*\otimes \Omega^0(\gamma)$ and $\mb{p}_\z=0$,  an observable, generalizing (\ref{Wilson B surface in abBF}) in non-abelian case.

\section{Final remarks}

In this paper we presented a 
two-step construction of classical observables in AKSZ sigma models, which consists of: (i) constructing an extension of the AKSZ theory to a BV theory on a larger space of fields using an extension of the target to a Hamiltonian $Q$-bundle, (ii) using the BV push-forward to the original space of fields to produce an observable. We also provided some examples, and in particular recovered the well-known Wilson loop observable in Chern-Simons theory together with its path integral representation \cite{AFS} and the Cattaneo-Rossi ``Wilson surface'' observable in $BF$ theory.

There are natural questions to this construction not answered here, to which we hope to return in a future publication:
\begin{itemize}
\item Extend the zoo of explicit examples. In particular,
\begin{itemize}
\item give more explicit examples of ``Wilson loops'' in Poisson sigma model, section \ref{sec: Wilson loop in PSM} (which requires a solution of (\ref{PSM Wilson loop eq on F}) as the input for step (i) of the construction and its quantization satisfying (\ref{PSM Wilson loop eq on F hat}) for step (ii)) other than those mentioned in remark \ref{rem: PSM Wilson loop};
\item give an example corresponding to a \textit{non-trivial} Hamiltonian $Q$-bundle over the target of the AKSZ sigma model.
\end{itemize}
\item Calculate the isotopy invariants of embeddings given by expectation values of our observables.
\end{itemize}

\subsection{Extension to source manifolds with boundary.}
Our construction of observables possesses an
extension to the case when the source manifold $\Sigma$ has a boundary $\dd \Sigma$, and the submanifold $\gamma\subset \Sigma$ on which the observable is supported is also allowed to have boundary, $\dd \gamma\subset \dd \Sigma$. We will outline this extension here, treating gauge theories with boundary in ``BV-BFV formalism'' \cite{CMR}. For a more detailed example, Chern-Simons theory on a manifold with boundary with Wilson lines ending on the boundary, the reader is referred to \cite{ABM}.

Recall \cite{CMR} that a gauge theory on manifold $\Sigma$ with boundary $\dd\Sigma$ in BV-BFV formalism is described by the following data:
\begin{itemize}
\item A degree 0 Hamiltonian $Q$-manifold $(\FF_{\dd\Sigma},Q_{\dd\Sigma},\Omega_{\dd\Sigma}=\delta\alpha_{\dd\Sigma},S_{\dd\Sigma})$ --- ``the BFV phase space'' (or ``the space of boundary fields'') associated to the boundary $\dd\Sigma$.
\item A $Q$-manifold $(\FF_\Sigma,Q_\Sigma)$ (``the space of bulk fields'') endowed with a $Q$-morphism $\pi: \FF_\Sigma\ra \FF_{\dd\Sigma}$ (restriction of fields to the boundary), a degree -1 symplectic form $\Omega_\Sigma$ and a degree 0 action $S_\Sigma\in C^\infty(\FF_\Sigma)$ satisfying
    \be \delta S_\Sigma =\iota_{Q_\Sigma}\Omega_\Sigma+ \pi^* \alpha_{\dd\Sigma} \label{BV-BFV eq}\ee
\end{itemize}
Equation (\ref{BV-BFV eq}) replaces the condition that $S_\Sigma$ is the Hamiltonian function for $Q_\Sigma$ in definition \ref{def: class BV theory}, and its consequence $Q_\Sigma S_\Sigma = \pi^*(\iota_{Q_{\dd\Sigma}}\alpha_{\dd\Sigma}- 2 S_{\dd\Sigma})$ or, equivalently,
$$\frac{1}{2}\underbrace{\{S_\Sigma,S_\Sigma\}_{\Omega_\Sigma}}_{:=\iota_{Q_\Sigma}\iota_{Q_\Sigma}\Omega_\Sigma}+ \pi^* S_{\dd\Sigma}=0$$
replaces the classical master equation.

AKSZ sigma models on manifolds with boundary fit naturally into BV-BFV formalism: one takes the standard construction (\ref{mapping space},\ref{Q AKSZ},\ref{AKSZ Omega},\ref{AKSZ S}) for the bulk and sets
\begin{multline*}
\FF_{\dd\Sigma}=\Map(T[1]\dd\Sigma,\MM),\quad Q_{\dd\Sigma}=d_{\dd\Sigma}^\lifted+Q^\lifted,\quad \Omega_{\dd\Sigma}= \tau_{\dd\Sigma}(\omega), \\
\alpha_{\dd\Sigma}=(-1)^{\dim\Sigma-1}\tau_{\dd\Sigma}(\alpha),\quad S_{\dd\Sigma}=(-1)^{\dim\Sigma-1}\left(\iota_{d_{\dd\Sigma}^\lifted}\tau_{\dd\Sigma}(\alpha)+ \tau_{\dd\Sigma}(\Theta)\right)
\end{multline*}
for the boundary; $Q$-morphism $\pi:\Map(T[1]\Sigma,\MM)\ra \Map(T[1]\dd\Sigma,\MM)$ is the pull-back by the inclusion $\dd\Sigma\hra\Sigma$.

A classical observable for a gauge theory in BV-BFV formalism (generalizing definition \ref{def: class observable}) can be defined as the following collection of data.
\begin{itemize}
\item A graded vector space (the auxiliary space of states) $\HH_{\dd\gamma}$ and a degree 1 element
$$\ddd_{\dd \gamma}\in C^\infty(\FF_{\dd\Sigma})\otimes \End(\HH_{\dd\gamma})$$
satisfying
$$Q_{\dd\Sigma} \ddd_{\dd\gamma}+\ddd_{\dd\gamma}^2=0$$
\item A degree 0 $\HH_{\dd\gamma}$-valued function of bulk fields
$$O_{\gamma}\in C^\infty(\FF_\Sigma)\otimes \HH_{\dd\gamma}$$
satisfying
$$(Q_\Sigma+\pi^*\ddd_{\dd\gamma})O_\gamma=0 $$
\end{itemize}
The locality requirement is that $O_\gamma$ only depends on the restrictions of bulk fields to $\gamma$; likewise, $\ddd_{\dd \gamma}$ should only depend on the restrictions of boundary fields to $\dd\gamma$.

We can also define a pre-observable in the BV-BFV context as the following data.
\begin{itemize}
\item A degree 0 Hamiltonian $Q$-bundle over $(\FF_{\dd\Sigma},Q_{\dd\Sigma})$ with fiber data $(\FF_{\dd\gamma},\A_{\dd\gamma},\Omega_{\dd\gamma}=\delta\alpha_{\dd\gamma},S_{\dd\gamma})$.
\item A trivial $Q$-bundle over $(\FF_\Sigma,Q_\Sigma)$ with fiber data $(\FF_{\gamma},\A_\gamma)$, with projection $\bar\pi:\FF_\gamma\ra \FF_{\dd\gamma}$ such that $d\bar\pi( \A_\gamma)=\A_{\dd\gamma}$, with fiber degree -1 symplectic structure $\Omega_\gamma$ and with degree 0 auxiliary action $S_\gamma\in C^\infty(\FF_\Sigma\times \FF_\gamma)$ satisfying
    $$\delta^\ver S_\gamma=\iota_{\A_\gamma}\Omega_\gamma+\bar\pi^*\alpha_{\dd\gamma},\quad
    Q_\Sigma S_\gamma+\frac{1}{2}\underbrace{\{S_\gamma,S_\gamma\}_{\Omega_\gamma}}_{:=\iota_{\A_\gamma}\iota_{\A_\gamma}\Omega_\gamma}+ \pi_\tot^*S_{\dd\gamma} = 0
    $$
    where $\delta^\ver$ is the de Rham differential on $\FF_\gamma$ and $\pi_\tot=\pi\times\bar\pi: \FF_\Sigma\times \FF_\gamma\ra \FF_{\dd\Sigma}\times \FF_{\dd\gamma}$.
\end{itemize}

Passing from a pre-observable to observable is a quantization problem. We construct $\HH_{\dd\gamma}$ as the geometric quantization of the symplectic manifold $(\FF_{\dd\gamma},\Omega_{\dd\gamma})$ and $\ddd_{\dd\gamma}$ as the geometric quantization of $S_{\dd\gamma}$:
\be \HH_{\dd\gamma}=\mr{GeomQuant}(\FF_{\dd\gamma},\Omega_{\dd\gamma}),\qquad \ddd_{\dd\gamma}=i\hat S_{\dd\gamma} \label{H via geomQ}\ee
For simplicity, assume that the Lagrangian polarization of $\FF_{\dd\Sigma}$ used for the geometric quantization is the vertical polarization of a fibration $\p:\FF_{\dd\Sigma}\ra \BB$. Then $O_\gamma$ is given by the following modification of the fiber BV integral:
\be O_\gamma(X,b)=\int_{\LL_b\subset \bar\pi^{-1}\p^{-1}b}\DD Y\; e^{i S_\gamma(X,Y)} \label{O_gamma via fiber BV int}\ee
for $b\in \BB$, $X\in\FF_\Sigma$; $Y$ runs over $\FF_\gamma$; $\LL_b$ is a Lagrangian in $\bar\pi^{-1}\p^{-1}b\subset \FF_\gamma$.

Given an AKSZ theory on $\Sigma$ with target $\MM$ and a Hamiltonian $Q$-bundle over $\MM$, as a first step, we construct a pre-observable with the bulk data given by old formulae (\ref{AKSZ pre-observable}) and boundary data given by the same formulae with $\gamma$ replaced by $\dd\gamma$:
\begin{multline*}
\FF_{\dd\gamma}=\Map(T[1]\dd\gamma,\NN),\quad \A_{\dd\gamma}=d_{\dd\gamma}^\lifted+p_\dd^*\A^\lifted,\quad
\Omega_{\dd\gamma}=\tau_{\dd\gamma}(\omega'),\\
\alpha_{\dd\gamma}=(-1)^{\dim\gamma-1}\tau_{\dd\gamma}(\alpha'), \quad S_{\dd\gamma}=(-1)^{\dim\gamma-1}\left(\iota_{d_{\dd\gamma}^\lifted}\alpha_{\dd\gamma}+p_\dd^*\tau_{\dd\gamma}^\tot(\Theta')\right)
\end{multline*}
where $p_\dd:\Map(T[1]\dd\Sigma,\MM)\ra \Map(T[1]\dd\gamma,\MM)$ is the restriction of ambient boundary fields to $\dd\gamma$.

As the second step, we pass from this pre-observable to an observable using construction (\ref{H via geomQ},\ref{O_gamma via fiber BV int}).

\end{document}